\documentclass[11pt,reqno,twoside]{article}

\usepackage{fixltx2e} %

\usepackage{cmap} %

\usepackage[T1]{fontenc}
\usepackage[utf8]{inputenc}
\usepackage{graphicx}
\usepackage{placeins}
\usepackage{enumerate}
\usepackage{algcompatible}

\usepackage{subcaption}

\usepackage{verbatim}

\usepackage{setspace}

\let\counterwithin\relax  %
\usepackage{lmodern} %
\usepackage[scale=0.88]{tgheros} %

\newcommand{\D}{\mathcal{D}}

\usepackage{bm} %

\usepackage{amsmath,amsbsy,amsgen,amscd,amsthm,amsfonts,amssymb} 

\usepackage[centering,top=1.5in,bottom=1.2in,left=1in,right=1in]{geometry}

\usepackage{titling}
\usepackage{cases}

\graphicspath{ {./images/} }

\usepackage[sf,bf,compact]{titlesec}

\usepackage{booktabs,longtable,tabu} %
\usepackage[font=small,margin=12pt,labelfont={sf,bf},labelsep={space}]{caption}

\usepackage[usenames,dvipsnames]{xcolor}

\definecolor{dark-gray}{gray}{0.3}
\definecolor{dkgray}{rgb}{.4,.4,.4}
\definecolor{dkblue}{rgb}{0,0,.5}
\definecolor{medblue}{rgb}{0,0,.75}
\definecolor{rust}{rgb}{0.5,0.1,0.1}

\usepackage{url}
\usepackage[colorlinks=true]{hyperref}
\hypersetup{linkcolor=dkblue}    
\hypersetup{citecolor=rust}      
\hypersetup{urlcolor=rust}     

\usepackage[final]{microtype}

\newtheoremstyle{myThm} %
    {\topsep}                    %
    {\topsep}                    %
    {\itshape}                   %
    {}                           %
    {\sffamily\bfseries}                   %
    {.}                          %
    {.5em}                       %
    {}  %

\newtheoremstyle{myRem} %
    {\topsep}                    %
    {\topsep}                    %
    {}                   %
    {}                           %
    {\sffamily}                   %
    {.}                          %
    {.5em}                       %
    {}  %

\newtheoremstyle{myDef} %
    {\topsep}                    %
    {\topsep}                    %
    {}                   %
    {}                           %
    {\sffamily\bfseries}                   %
    {.}                          %
    {.5em}                       %
    {}  %

\theoremstyle{myThm}
\newtheorem{theorem}{Theorem}[section]
\newtheorem{lemma}[theorem]{Lemma}
\newtheorem{proposition}[theorem]{Proposition}

\theoremstyle{myRem}
\newtheorem{remark}[theorem]{Remark}

\theoremstyle{myDef}

\usepackage{algorithm}
\usepackage{algpseudocode}

\let\originalleft\left
\let\originalright\right
\renewcommand{\left}{\mathopen{}\mathclose\bgroup\originalleft}
\renewcommand{\right}{\aftergroup\egroup\originalright}

\usepackage{mathtools}
\mathtoolsset{centercolon}  %

\renewcommand{\phi}{\varphi}
\newcommand{\eps}{\varepsilon}

\providecommand{\mathbbm}{\mathbb} %

\newcommand{\N}{\mathbbm{N}}

\newcommand{\G}{\mathcal{G}}

\definecolor{mygreen}{rgb}{0.1,0.75,0.2}

\newcommand{\Nc}{\mathcal{N}}

  \newcommand{\0}{{\bf{0}}}
\newcommand{\cmat}{c_{\mbox {\tiny{\rm Mat}}}}
\renewcommand{\i}{ {\bf {i}}}  
\newcommand{\x}{ {\bf {x}}}  
\newcommand{\X}{ {\bf {X}}}  
\newcommand{\y}{ {\bf {y}}}  
\newcommand{\w}{ {\bf {w}}}  
\newcommand{\z}{ {\bf {z}}}  
\renewcommand{\k}{ {\bf {k}}}  
 \usepackage[normalem]{ulem}
  \newcommand{\Q}{{\bf{Q}}}
 \renewcommand{\G}{{\bf{G}}}
  \newcommand{\M}{{\bf{M}}}
  \renewcommand{\S}{{\bf{S}}}
    \newcommand{\R}{{\bf{R}}}
   \newcommand{\I}{{\bf{I}}}
    
\newcommand{\bLambda}{{\bf{\Lambda}}}
\newcommand{\h}{ {\bf {h}}}

\usepackage[]{algorithm}

\usepackage{graphicx}

\usepackage{soul}
\usepackage{authblk}
\usepackage[square,numbers]{natbib}
\usepackage{chngcntr}
\usepackage{mathrsfs}
\counterwithin{table}{section}
\counterwithin{algorithm}{section}

\title{Finite Element Representations of Gaussian Processes: \\ Balancing Numerical and Statistical Accuracy}   %
\author{Daniel Sanz-Alonso and Ruiyi Yang}

\vspace{.25in} 

\date{University of Chicago}

\makeatletter\@addtoreset{section}{part}\makeatother%
\numberwithin{equation}{section}

\newcommand{\upperRomannumeral}[1]{\uppercase\expandafter{\romannumeral#1}}

\graphicspath{ {./images/} }

\begin{document}

\maketitle

\begin{abstract}
  The stochastic partial differential equation approach to Gaussian processes (GPs) represents Mat\'ern  GP priors in terms of $n$ finite element basis functions and Gaussian coefficients with sparse precision matrix. 
Such representations enhance the scalability of GP regression and classification to datasets of large size $N$ by setting $n\approx N$ and exploiting sparsity. 
In this paper we reconsider the standard choice $n \approx N$ through an analysis of the estimation performance.
Our theory implies that, under certain smoothness assumptions, one can reduce the computation and memory cost without hindering the estimation accuracy by setting $n \ll N$ in the large $N$ asymptotics. Numerical experiments illustrate the applicability of our theory and the effect of the prior lengthscale in the pre-asymptotic regime.
\end{abstract}

\section{Introduction}
Gaussian processes (GPs) are an important model for prior distributions over functions, and play a central role in spatial statistics, machine learning, Bayesian inverse problems, and a variety of other scientific and engineering applications \cite{stein1999interpolation,williams2006gaussian,stuart2010inverse,sullivan2015introduction,owhadi2019statistical,owhadi2019operator}. However, GP methodology often suffers from the \emph{big $N$ problem}: conditioning a GP to $N$ observations requires to factorize an $N\times N$ covariance matrix, with a general cost of $O(N^3)$. Numerous approaches to address this challenge have been developed \cite{heaton2019case}. The aim of this paper is to provide novel understanding of the popular stochastic partial differential equation (SPDE) approach   \cite{lindgren2011explicit,lindgren2022spde}   for GP regression and classification with large datasets. 

Let $u(\x)$ be a Mat\'ern-type GP (see e.g. \cite{bolin2020numerical} or Subsection \ref{ssec:SPDEbg} below) on a bounded domain $\D \subset \mathbb{R}^D.$ The SPDE approach approximates $u$ with a GP $u_h$ of the form
\begin{equation}\label{eq:GPexpansion}
u_h(\x)=\sum_{i=1}^{n_h} w_i e_i(\x), \quad \quad \x \in \D ,
\end{equation}
where $e_i: \D \to \mathbb{R}$ are finite element (FE) basis functions and $\w := (w_1, \ldots, w_{n_h})^\top  \sim \mathcal{N}(\0,\Q^{-1})$ with  sparse precision matrix $\Q \in \mathbb{R}^{n_h \times n_h}.$ The dimension $n_h$ of the basis $\{e_i\}_{i=1}^{n_h}$ is determined by a mesh-size parameter $h>0.$ Previous work sets $h$ so that $n_h \approx N,$ and the $O(N^3)$ computational cost is reduced by exploiting the local support of the FE basis functions $e_i$ and the sparsity of $\Q$,  see e.g.  \cite{lindgren2011explicit,bolin2013comparison,bolin2020rational} or Subsection \ref{ssec: FEM for regression and classification}. However, the choice $n_h \approx N$ has not been theoretically or empirically investigated. In particular, it is not clear if the computational gain achieved with $n_h\approx N$ comes at the price of  larger  estimation error. In this paper we shall introduce a framework for selecting $n_h$ based on the posterior estimation performance achieved when using GP prior $u_h$. 
Our theory implies that, under certain smoothness assumptions, choosing $n_h \ll N$ can indeed be sufficient in the large $N$ asymptotics, as otherwise the statistical errors inherent to the regression or classification tasks dominate the numerical error in the approximation $u_h \approx u$. Therefore, in addition to the computational gain facilitated by sparsity, there is a second computational and memory gain: the \emph{dimension} of the matrices that need to be factorized can be reduced in large $N$ regimes without hindering the estimation accuracy. Numerical experiments will illustrate the applicability of our theory and the effect of the prior lengthscale in the pre-asymptotic regime.
  
The SPDE approach is part of a trend in GP methodology that seeks to leverage sparsity for computational efficiency   \cite{quinonero2005unifying}.   In this spirit, one can construct sparse approximations of the covariance matrix of the observations (a procedure known as \emph{tapering} or \emph{localization} \cite{gaspari1999construction,furrer2006covariance}),    or of the precision matrix \cite{datta2016hierarchical} and its Cholesky factor \cite{schafer2021sparse,kang2021correlation}. Other approaches exploiting sparsity include  Vecchia approximations \cite{vecchia1988estimation,katzfuss2021general} and methods based on the screening effect \cite{stein2002screening}.   These techniques are well established in several applications and are essential,   for example,   in the practical implementation of data assimilation algorithms for numerical weather forecasting \cite{houtekamer2001sequential}.
 A complementary line of work relies on \emph{smoothness} rather than sparsity for computational expediency. For instance, truncated Karhunen-Lo\`eve expansions in Bayesian inverse problems rely on a representation of the form  \eqref{eq:GPexpansion} with small dimension $n$, spectral basis functions, and stochastic weights with diagonal covariance \cite{stuart2010inverse}. These low-rank representations \cite{solin2020hilbert,greengard2021efficient} have   been claimed to remove   fine-scale variations of the process \cite{bolin2020rational},   but can be accurate if the underlying process is smooth.   Our work blends sparsity and smoothness demonstrating that, for regression and classification with large data-sets, sparse methods can benefit from a significant dimension reduction under mild smoothness assumptions.

To propose a criterion for choosing $n_h$ with respect to $N$, we will exploit the concept of posterior contraction rates \cite{ghosal2000convergence}, which is discussed in Subsection \ref{ssec:bg pc}. Roughly speaking, we consider a scaling sufficient if the posterior constructed with GP prior $u_h$ contracts at the same rate as the posterior constructed with the true GP prior $u.$  
The Bayesian nonparametrics framework in \cite{van2008rates} guarantees that if the rate of convergence of the GP prior approximation $u_h \approx u$ is fast enough, then the corresponding posteriors contract at the same rate. Establishing convergence rates for approximations $u_h \approx u$ is an active research area  on numerical analysis of FE solution of fractional SPDEs \cite{bolin2020numerical,bolin2020rational,cox2020regularity}. As part of our analysis, we derive a crude estimate of the approximation error $\mathbb{E}\|u_h-u\|_{\infty}^2$ for a particular FE discretization when $\mathcal{D}$ is a hyperrectangle.
The result holds for general dimension $D$ while being less sharp than the one-dimensional result in \cite{cox2020regularity}, which also allows for more general domains. 
However, our main objective is to illustrate that the plug-in character of the framework \cite{van2008rates} allows to seamlessly translate $L^\infty$ and $L^2$ error bounds for the approximation $u_h  \approx u$ into sufficient choices of $n_h$ in terms of $N$ in regression and classification settings. As we shall see, even crude error bounds suggest that, in the large $N$ asymptotics, $n_h \ll N$ can be sufficient under mild smoothness assumptions.

Numerical simulations in the regression setting will complement our theoretical analysis. Our experiments illustrate that (i) the qualitative theoretical behavior suggested by our large $N$ asymptotic analysis is in agreement with the behavior observed with moderate sample-size;  (ii) if the truth is not smooth and has a short lengthscale, choosing $n_h \gg N$ may indeed be necessary for the SPDE approach to match the estimation accuracy of the ground truth prior model; and (iii) outside the large $N$ asymptotic regime, the prior lengthscale plays an important role in determining appropriate choice of $n_h$ in terms of $N.$ This last point is also partly explained by our theory, where the lengthscale appears as a prefactor in the error bound for the FE prior representations. We believe these findings together with our theoretical results provide useful insights for calibrating the FE approach in practice.

The study of fixed-domain, large $N$ asymptotics \cite{stein1999interpolation,stein1990uniform,stein1990bounds,stein1999predicting,du2009fixed,wang2011fixed}  is motivated by applications in environmental science, ecology, climate, and hydrology, where $N$ is often in the order of hundreds of thousands or larger. At a high level, our criterion resembles the in-fill asymptotic analysis of tapered covariance functions in \cite{furrer2006covariance}, where the authors give conditions on the taper function that guarantee large-data asymptotic equivalence of the mean-squared prediction error of the true and tapered covariance models. As in \cite{furrer2006covariance}, we may interpret $u_h$ as defining a misspecified  covariance model and then, similar to \cite{furrer2006covariance,stein1993simple}, our criterion guarantees that the misspecification is inconsequential in a large-data regime. On the other hand, even if the Mat\'ern-type GP $u$ is not interpreted as a ground truth prior model, our analysis suggests that over-discretizing the FE representations $u_h$ should be avoided, as there is a threshold beyond which further discretizing increases the computational cost without improving the estimation accuracy. 
Similar ideas permeate the study of the value of unlabeled data in  semi-supervised learning  \cite{sanz2020unlabeled}  with graph representations of Mat\'ern GPs \cite{sanz2020spde}.

The rest of this paper is organized as follows. We provide all necessary background and formalize our problem setting in Section \ref{sec:bg}. Our main results are in Section \ref{sec:main}   and complementary   numerical experiments in Section \ref{sec:numerics}. We close in Section \ref{sec:conclusion} with    possible extensions of our main results and open directions that stem from our work.    All the proofs are deferred to Section \ref{sec:proofs}.

\paragraph{Notation.}
For $a,b$ two real numbers, we denote $a\wedge b=$ min$\{a,b\}$ and $a\vee b=$ max$\{a,b\}$. The symbol $\lesssim$ will denote less than or equal to up to a universal constant. For two real sequences $\{a_i\}$ and $\{b_i\}$, we denote   (i) $a_i\ll b_i$ if $\operatorname{lim}_i (a_i/b_i)=0$;    (ii) $a_i=O(b_i)$ if $\operatorname{lim\, sup}_i (a_i/b_i)\leq C$ for some positive constant $C$; and  (iii) $a_i\asymp b_i$ if  $c_1\leq \operatorname{lim\,inf}_i (a_i/b_i) \leq \operatorname{lim\,sup}_i (a_i/b_i) \leq c_2$ for some positive constants $c_1,c_2$. For a nonnegative integer $K$, we denote $[K]=\{0,\ldots,K\}$.

\section{Background and Problem Setting}\label{sec:bg}
To make our presentation self-contained, we introduce in this section all necessary background and formalize our problem setting. Mat\'ern-type GPs and their connection with the classical Mat\'ern covariance function are discussed in Subsection \ref{ssec:SPDEbg}. Subsection \ref{ssec:bg fem} reviews FE representations of Mat\'ern-type GPs. Our regression and classification problem settings are formalized in
 Subsection \ref{ssec: FEM for regression and classification}, where we also summarize how FE representations of Mat\'ern-type GPs allow to speed up computations. Finally, Subsection \ref{ssec:bg pc} overviews the Bayesian nonparametrics framework that we employ as our criterion to identify sufficient scalings of $n_h$ with respect to $N.$

\subsection{The Mat\'ern Covariance Function and SPDE Representations}\label{ssec:SPDEbg}
Recall that the Mat\'ern covariance function is defined by
\begin{align}
    \cmat(\x,\x')=\sigma^2 \frac{2^{1-\nu}}{\Gamma(\nu)}\left(\kappa|\x-\x'|\right)^{\nu}K_{\nu}\left(\kappa|\x- \x'|\right), \quad \quad \x,\x'\in\mathbb{R}^D,\label{eq:matern covariance}
\end{align}
where $|\cdot|$ is the Euclidean distance on $\mathbb{R}^D$, $\Gamma$ is the gamma function and $K_{\nu}$ is the modified Bessel function of the second kind. The parameters $\sigma$, $\nu$, $\kappa$ control, respectively, the marginal variance, smoothness of the sample paths, and correlation lengthscale.
Due to its flexibility, the Mat\'ern model is widely used in spatial statistics \cite{stein1999interpolation,gelfand2010handbook}, machine learning \cite{williams2006gaussian}, and uncertainty quantification \cite{sullivan2015introduction}, with applications in various scientific fields \cite{guttorp2006studies,cameletti2013spatio}.
The connection between the Mat\'ern covariance and SPDEs has long been noticed \cite{whittle1954stationary}. Consider formally the equation
\begin{align}
    (\kappa^2-\Delta)^{s/2}u = \kappa^{s-D/2}\mathcal{W}  \quad \quad \text{in}\,\, \mathcal{D}, \label{eq:spde global}
\end{align}
where $s= \nu + D/2,$ $\Delta$ is a Laplacian and $\mathcal{W}$ is a spatial white noise. (Here and below we will ignore the marginal variance which acts only as a scaling factor.) If $\D : = \mathbb{R}^D,$ then the unique stationary solution to \eqref{eq:spde global}, suitably interpreted \cite{whittle1954stationary}, has covariance function \eqref{eq:matern covariance}. 

Following \cite{lindgren2011explicit}, we will define \emph{Mat\'ern-type GPs}  by solution of  \eqref{eq:spde global} in a \emph{bounded} domain $\D \subset \mathbb{R}^D$, interpreting the SPDE \eqref{eq:spde global}  as in \cite{bolin2020numerical}. We outline here the main ideas and refer to \cite{bolin2020numerical} for further details. Let $\mathcal{L} :=\kappa^2-\Delta$ be equipped with homogeneous Dirichlet or Neumann boundary condition. The eigenfunctions $\{\psi_i\}_{i=1}^{\infty}$   of the Dirichlet (or Neumann) Laplacian   form an orthonormal basis of $L^2(\mathcal{D})$, where the associated ordered eigenvalues $\{\lambda_i\}_{i=1}^{\infty}$ satisfy $\lambda_i\asymp i^{2/D}$ by Weyl's law (see e.g.  \cite[Theorem 6.3.1]{davies1996spectral}).  The fractional power operator $\mathcal{L}^{s/2}$ in \eqref{eq:spde global} is then defined by
\begin{align*}
    \mathcal{L}^{s/2} u :=\sum_{i=1}^{\infty} (\kappa^2+\lambda_i)^{s/2}\langle u,\psi_i \rangle \psi_i 
\end{align*}
with domain $\left\{u\in L^2(\mathcal{D}): \sum_{i=1}^{\infty}(\kappa^2+\lambda_i)^{s}\langle u,\psi_i \rangle^2<\infty\right\}$, where $\langle \cdot,\cdot\rangle$ denotes the $L^2(\mathcal{D})$-inner product. The white noise in \eqref{eq:spde global} is  formally defined by the series $\mathcal{W} = \sum_{i=1}^\infty \xi_i \psi_i,$ with $\xi_i \smash{\overset{i.i.d.}{\sim}} \mathcal{N}(0,1) $ set on a complete probability space $(\Omega,\mathcal{A},\mathbb{P}).$ As rigorously shown in \cite[Lemma 2.1]{bolin2020numerical},
existence and uniqueness of solutions to \eqref{eq:spde global} in $L^2(\Omega;L^2(\mathcal{D}))$ is guaranteed for $s>D/2$. Moreover, the solution can be represented as a series expansion 
\begin{align}
    u(\x)=\kappa^{s-D/2}\sum_{i=1}^{\infty} (\kappa^2+\lambda_i)^{-s/2}\xi_i \psi_i(\x),\quad \quad \xi_i\overset{i.i.d.}{\sim} \mathcal{N}(0,1), \quad \quad \x \in \mathcal{D}, \label{eq:KL expansion true field}
\end{align}
where the assumption $s>D/2$ together with Weyl's law guarantees that $u\in L^2(\mathcal{D})$ almost surely. 
We refer to $u$ defined by \eqref{eq:KL expansion true field} as a Mat\'ern-type GP.
The covariance function of Mat\'ern-type GPs  no longer agrees with the classical Mat\'ern covariance model \eqref{eq:matern covariance}, but approximates it well away from the boundary ---see for instance Proposition \ref{prop:covariances are close} below.

\subsection{Finite Element Representations of Mat\'ern-type Gaussian Processes}\label{ssec:bg fem}
Let $\mathcal{D} \subset \mathbb{R}^D$ be a bounded domain and let $\{V_h\}_{h\in(0,1)}$ be a family of subspaces of $H^1(\mathcal{D})$ (the space of functions whose weak derivatives belong to $L^2(\mathcal{D})$) with finite dimensions $n_h : = \text{dim} (V_h)<\infty$. In subsequent developments $h$ will play the role of a mesh-size parameter and $n_h \asymp  h^{-D}$. Consider the Galerkin discretization $-\Delta_h:V_h\rightarrow V_h$ of $-\Delta$  defined as 
\begin{align*}
    \langle -\Delta_h u_h,v_h \rangle = \langle -\Delta u_h,v_h\rangle \quad \quad \forall \, u_h,v_h\in V_h.
\end{align*}
Let $\{(\lambda_{h,i},\psi_{h,i})\}_{i=1}^{n_h}$ be the eigenpairs of $-\Delta_h$ satisfying
\begin{align*}
    \langle -\Delta_h \psi_{h,i}, v_h\rangle = \lambda_{h,i}\langle  \psi_{h,i}, v_h\rangle \quad \quad \forall v_h\in V_h,
\end{align*}
where we assume the $\lambda_{h,i}$'s are in increasing order and the $\psi_{h,i}$'s are orthonormal.  We then define a discretization of the SPDE  \eqref{eq:spde global} by
\begin{align}
    \mathcal{L}_{h}^{s/2}u_h:=(\kappa^2-\Delta_h)^{s/2}u_h =\kappa^{s-D/2}\mathcal{W}_h, \quad \quad \mathcal{W}_h:=\sum_{i=1}^{n_h} \xi_i \psi_{h,i},\quad \quad \xi_i \overset{i.i.d.}{\sim} \mathcal{N}(0,1). \label{eq:SPDE fem}
\end{align}
We refer to the solution $u_h$ as a FE representation of the Mat\'ern-type GP $u$. Note that 
\begin{equation}
u_h(\x) =\kappa^{s-D/2}\sum_{i=1}^{n_h}(\kappa^2+\lambda_{h,i})^{-s/2}\xi_i \psi_{h,i}(\x), \quad \quad \xi_i \overset{i.i.d.}{\sim} \mathcal{N}(0,1), \quad \quad \x \in \mathcal{D}. \label{eq:KL fem approx}
\end{equation}
Inspection of \eqref{eq:KL expansion true field} and \eqref{eq:KL fem approx} suggests that the error in the approximation $u_h \approx u$ is largely determined by the FE error in the approximations $\lambda_{h,i} \approx \lambda_i$ and $\psi_{h,i} \approx \psi_i.$ We will pursue this idea in our error analysis in Section \ref{sec:main}. However, the Karhunen-Lo\`eve representation \eqref{eq:KL fem approx} is not   in general    useful for practical implementation, as the eigenpairs $\{(\lambda_{h,i},\psi_{h,i})\}_{i=1}^{n_h}$   can be   expensive to compute and the eigenfunctions do not have compact support.   The following result from \cite{lindgren2011explicit}   shows that the solution to \eqref{eq:SPDE fem} admits an equivalent representation in terms of a FE basis, as foreshadowed in \eqref{eq:GPexpansion}. 
\begin{proposition}\label{prop:covariance of FEM approximation}
Let $\{e_{h,i}\}_{i=1}^{n_h}$ be a FE basis of $V_h,$ and denote by $\M$ and $\G$ the mass and stiffness matrices with entries $\M_{ij}=\langle e_{h,i},e_{h,j} \rangle$ and $\G_{ij}=\langle \nabla e_{h,i},\nabla e_{h,j}\rangle$.  
For $0\neq s \in \N,$ the FE representation $u_h$ of the Mat\'ern-type GP $u$ admits the characterization
\begin{align}
 u_h(\x) =\sum_{i=1}^{n_h}w_i e_{h,i}(\x), \quad   \w\sim\mathcal{N}(\0,\Q^{-1}),        \label{eq: fem repre}
\end{align}
where $\Q=(\kappa^2\M+\G)\bigl[\M^{-1}(\kappa^2 \M+\G)\bigr]^{s-1}.$
\end{proposition}

Notice that \eqref{eq: fem repre} does not involve the eigenpairs. Moreover,
the matrices $\M$ and $\G$ are sparse for standard FE basis $e_{h,i}$, e.g. tent functions. Lumping the mass matrix $\M$ ensures sparsity of $\Q$ and gives a Gauss-Markov approximation to the Mat\'ern-type GP $u$ \cite{lindgren2011explicit}.  For $s\notin\mathbb{N}$,  the rational SPDE approach can be adopted \cite{bolin2020rational}.

\subsection{Gaussian Process Regression and Classification: Finite Element Representations}\label{ssec: FEM for regression and classification}
Here we introduce the regression and classification models we consider, and describe briefly how  FE representations of Mat\'ern-type GPs can alleviate the computational burden of these tasks.
Given $N$ pairs of data $\{(\X_i,Y_i)\}_{i=1}^N$ we are interested in inferring  $f_0(\x)=\mathbb{E}[Y|\X=\x]$ under the following data-generating mechanisms:
\begin{itemize}
    \item Fixed design regression: $Y_i=f_0(\X_i)+\eta_i$, where the $\X_i$'s are fixed   (and distinct) 
    covariates   and $\eta_i\overset{i.i.d.}{\sim}\mathcal{N}(0,\tau^2)$ with $\tau$ known.
    \item Binary classification: $\mathbb{P}(Y_i=1| \X_i)=f_0(\X_i)$, where $\X_i\overset{i.i.d.}{\sim} \mu$ for some distribution $\mu$ over $\mathcal{D}$.
\end{itemize}
For simplicity we shall assume for the rest of this paper that $\mu$ is the uniform distribution over $\mathcal{D},$ but we note that it suffices to assume that $\mu$ admits a Lebesgue density bounded above and below by positive constants.

For fixed design regression, we set a FE Mat\'ern-type GP prior $u_h$ on $f_0.$  The posterior of the weights $\w$ is given by
\begin{align*}
    \w|\{(\X_i,Y_i)\}_{i=1}^N \sim \mathcal{N}\big(    (\S^\top \S+\tau^{2}\Q)^{-1}\S^\top \y, (\tau^{-2}\S^\top \S+\Q)^{-1}\big),
\end{align*}
where $\S\in\mathbb{R}^{N\times n_h}$ has entries $\S_{ij} = e_j(\X_i)$ and $\y = (Y_1, \ldots, Y_N)^\top.$
The main computational cost for posterior inference is in factorizing the $n_h\times n_h$ matrix $\S^\top \S+\tau^2 \Q$. This factorization can be efficiently computed since the local support of standard FE basis functions ensures sparsity of $\S,$ and $\Q$ can be made sparse as discussed in Subsection \ref{ssec:bg fem}.

For binary classification, let $\Phi$ be the logistic function and consider a wrapped  GP prior $\Phi \circ u_h$ over $f_0$. 
The posterior log-density is given by
\begin{align}
    \log \mathbb{P}( \w| \{(\X_i,Y_i)\}_{i=1}^N) &= \sum_{i=1}^N Y_i\log \Phi((\S\w)_i) + (1-Y_i)\log (1-\Phi((\S\w)_i))\label{eq:binary classification posterio LL} \\
    &\quad \quad -\frac12 \w^\top \Q\w+ \text{const},\nonumber 
\end{align}
where $(\S\w)_i$ denotes the $i$-th entry of $\S\w$. Two standard procedures for posterior inference are \emph{maximum a posteriori } (MAP) estimation and Markov chain Monte Carlo (MCMC) sampling.
To compute the MAP estimate,  \eqref{eq:binary classification posterio LL} is optimized to recover the weights with highest posterior density. This optimization problem can be efficiently solved using the  Hessian of the objective function, which takes the form $\S^\top {\bf D} \S- \Q$, where ${\bf D}$ is a diagonal matrix with
\begin{align*}
  {\bf D}_{ii}= \Phi^{\prime\prime}((\S\w)_i)\left[\frac{Y_i}{\Phi((\S\w)_i)}-\frac{1-Y_i}{1-\Phi((\S\w)_i)}\right]-[\Phi^{\prime}((\S\w)_i)]^2\left[\frac{Y_i}{\Phi((\S\w)_i)^2}+\frac{1-Y_i}{(1-\Phi((\S\w)_i))^2}\right]. 
\end{align*}
Therefore, the computational cost is largely determined by the sparsity of the matrix $\S^\top{\bf D} \S-\Q$, which in turn depends on the sparsity of $\S$ and $\Q$.  
On the other hand, MCMC algorithms for posterior inference with GP priors have been widely studied \cite{bernardo1998regression,cotter2013,cui2016dimension,trillos2017consistency,sanzstuarttaeb}, and a key idea behind these methods is to employ a proposal mechanism $\w \mapsto \w'$ of the form 
\begin{equation}
\w' = \theta\, \w + (1-\theta)^{1/2} \boldsymbol{\gamma},  \quad \quad \boldsymbol{\gamma} \sim \Nc(0, \Q^{-1}),
\end{equation}
which leaves the prior distribution  $\Nc( \0, \Q^{-1})$ of the weights invariant. In order to sample $\boldsymbol{\gamma}\sim \Nc( \0, \Q^{-1})$ with large $n_h$ it is important to leverage sparsity of $\Q$ \cite{rue2005gaussian}.

\subsection{Our Criterion: Matching Posterior Contraction Rates}\label{ssec:bg pc}

The FE approach outlined above involves a user-chosen hyperparameter $h$ that affects both the estimation performance and computational cost. Smaller $h$ leads to better approximation of the Mat\'ern-type GP $u$ by $u_h$ and possibly enhanced inference, but renders a larger $n_h$  that increases the computational cost. Since $u_h$ is supposed to approximate the Mat\'ern-type GP $u$, a natural choice for $h$ is so that the estimation performance of using $u_h$ as the prior is ``comparable'' to that of $u$. In this section we shall formalize such intuition with the notion of posterior contraction rates. 

To begin with, recall that the goal is to infer the conditional expectation $f_0(\x)=\mathbb{E}[Y|\X=\x]$ from data $\{(\X_i,Y_i)\}_{i=1}^N$.   We shall adopt a frequentist Bayesian perspective by putting a sequence of priors $\Pi_N$ over $f_0$ and assuming that the data are indeed generated from a fixed $f_0$ which we interpret as the ground truth.   Following \cite{ghosal2000convergence}, we say that the sequence of posteriors with respect to $\Pi_N$ contracts around $f_0$ with rate $\varepsilon_N$ if, for any sufficiently large $M>0,$
\begin{align}
    \mathbb{E}_{f_0} \Pi_N\left(f:d_N(f,f_0)\leq M\varepsilon_N\,|\, \{(\X_i,Y_i)\}_{i=1}^N\right)\xrightarrow{N\rightarrow\infty} 1. \label{eq:pc def}
\end{align}
Here the expectation is taken with respect to the data distribution   of $\{(\X_i,Y_i)\}_{i=1}^N$ determined by $f_0$ and the marginal of the $\X_i$'s,   
and $d_N$ is a suitable discrepancy measure. Roughly speaking,   $\varepsilon_N$  is the rate at which one can shrink the radius of a ball centered around the truth while at the same time capturing almost all the posterior mass. The condition \eqref{eq:pc def} implies that asymptotically the sequence of posteriors will be nearly supported on a ball of radius $O(\varepsilon_N)$ around $f_0$.  Therefore, $\varepsilon_N$ can be loosely interpreted as the convergence rate of the posteriors towards the truth.   An important consequence  \cite[Theorem 2.5]{ghosal2000convergence} is that the point estimator defined as  
\begin{align*}
    \widehat{f}_N=\underset{g}{\operatorname{arg\,max}}\, \Big[\Pi_N\left(f:d_N(f,g)\leq M\varepsilon_N\,|\,\{(\X_i,Y_i)\}_{i=1}^N\right)\Big],
\end{align*}
converges   (in probability)   to $f_0$ with the same rate $\varepsilon_N$. Therefore the contraction rate serves as a natural criterion for quantifying the estimation performance of the posteriors.

Following Subsection \ref{ssec: FEM for regression and classification}, the sequence of priors is taken as $\Pi_N=\operatorname{Law}(u_{h_N})$ (resp. $\operatorname{Law}(\Phi(u_{h_N}))$) for fixed design regression (resp. binary classification). The selection criterion for $h_N$ that we propose is to choose $h_N$ so that the sequence of posteriors with respect to $\Pi_N$ contracts at the same rate as if $\Pi_N\equiv \Pi:=\operatorname{Law}(u)$ (resp. $\operatorname{Law}(\Phi(u))$),   where $u$ is the Mat\'ern-type GP that $u_{h_N}$ is approximating.   It turns out that there is a simple condition on the    approximation accuracy of $u_{h_N}$ towards $u$   that guarantees this matching of posterior contraction rates, which we make precise below.

  We start by reviewing the key ingredients of the theory when a single prior is adopted, i.e., when $\Pi_N\equiv \Pi$ in the above.  Consider now $u$ as a GP taking values in $(L^{\infty}(\mathcal{D}), \|\cdot\|_{\infty})$ (see e.g. Lemma \ref{lemma:infty field rate rectangle} below for conditions under which this is valid)   for fixed design regression  and in $(L^2(\mathcal{D}), \|\cdot\|_2)$ for binary classification. By \cite[Theorems 3.2 and 3.3]{van2008rates}, the contraction rate with respect to $\Pi$ in the fixed design regression (resp. binary classification) setting can be characterized as the sequence $\varepsilon_N$ that satisfies $\phi_{f_0}(\varepsilon_N; u, \|\cdot\|_{\infty})\leq N\varepsilon_N^2$ (resp. $\phi_{\Phi^{-1}(f_0)}(\varepsilon_N; u, \|\cdot\|_2)\leq N\varepsilon_N^2$), where 
\begin{align}
    \phi_{\omega_0}(\varepsilon; u,\|\cdot\|_{\mathbb{B}}):=\underset{g\in \mathbb{H}:\|g-\omega_0\|_{\mathbb{B}}< \varepsilon}{\operatorname{inf}}\, \|g\|^2_{\mathbb{H}}-\log \mathbb{P}(\|u\|_{\mathbb{B}}<\varepsilon), \label{eq:concentration function}
\end{align}
and $(\mathbb{H},\|\cdot\|_{\mathbb{H}})$ denotes the \emph{reproducing kernel Hilbert space} (RKHS) of $\Pi$ (see e.g. \cite{van2008reproducing} for more details). Under such circumstances, the sequence of posteriors with respect to $\Pi$ contracts around $f_0$ with rate $\varepsilon_N$ in the sense of \eqref{eq:pc def} with $d_N=\|\cdot\|_N$ the empirical norm defined as $\|f\|^2_N=N^{-1}\sum_{i=1}^N |f(\X_i)|^2$ for fixed design regression and $d_N=\|\cdot\|_2$ for binary classification. In other words, the posterior contraction rate can be determined by analyzing the so-called \emph{concentration function} \eqref{eq:concentration function} of the prior.   Now when a sequence of priors $\Pi_N$ is used instead,   it is reasonable to expect that if $\Pi_N$ approximates $\Pi$ sufficiently well, the concentration functions of $\Pi_N$ will be close to that of $\Pi$ so that the same contraction rate can be achieved. Indeed this is implied by \cite[Theorems 2.2, 3.2 and 3.3]{van2008rates}, which we record as a proposition. 
\begin{proposition}\label{prop:prior approx implies same pc rate}
\begin{enumerate}
    \item Fixed design regression: Let $\Pi_N=\operatorname{Law}(u_{h_N})$. Suppose $\varepsilon_N$ is a  sequence of real numbers satisfying $\phi_{f_0}(\varepsilon_N; u, \|\cdot\|_{\infty})\leq N\varepsilon_N^2$ and 
    \begin{align}
    10\mathbb{E}\|u_{h_N}-u\|_{\infty}^2\leq N^{-1}. \label{eq:prior approximation condition infty}
\end{align}
Then, for any sufficiently large $M>0,$
\begin{align*}
    \mathbb{E}_{f_0} \Pi_N\left(f:\|f-f_0\|_N\leq M\varepsilon_N\,|\, \{(\X_i,Y_i)\}_{i=1}^N\right)\xrightarrow{N\rightarrow\infty} 1.
\end{align*}
    \item Binary classification: Let $\Pi_N=\operatorname{Law}(\Phi(u_{h_N}))$. Suppose $\varepsilon_N$ is a sequence of real numbers satisfying $\phi_{\Phi^{-1}(f_0)}(\varepsilon_N; u, \|\cdot\|_2)\leq N\varepsilon_N^2$ and 
    \begin{align}
    10\mathbb{E}\|u_{h_N}-u\|_2^2\leq N^{-1}. \label{eq:prior approximation condition 2}
\end{align}
Then, for any sufficiently large $M>0,$
\begin{align*}
    \mathbb{E}_{f_0} \Pi_N\left(f:\|f-f_0\|_2\leq M\varepsilon_N\,|\, \{(\X_i,Y_i)\}_{i=1}^N\right)\xrightarrow{N\rightarrow\infty} 1.
\end{align*}
\end{enumerate}
\end{proposition}

Proposition \ref{prop:prior approx implies same pc rate} shows that the posteriors constructed with prior $\Pi_N$ and with prior $\Pi$ contract at the same rate, provided that the prior approximation is sufficiently accurate. 
Therefore it suffices to choose $h_N$ so that \eqref{eq:prior approximation condition infty} or \eqref{eq:prior approximation condition 2} is satisfied, giving a simple criterion for setting $h_N$.  In particular, if the error $\mathbb{E}\|u_{h_N}-u\|_{\infty}^2$ or $\mathbb{E}\|u_{h_N}-u\|_2^2$ decreases sufficiently fast, then a slowly decaying $h_N$ is enough and leads to $n_{h_N}\asymp h_N^{-D}\ll N$.   We will show in Section \ref{sec:main} for a simple linear FE method in a concrete setting that this is indeed the case under certain smoothness assumptions, and demonstrate such behavior through simulation studies in Section \ref{sec:numerics}. Several possible extensions will be discussed in Section \ref{sec:conclusion}, building on the key idea of using Proposition \ref{prop:prior approx implies same pc rate} to balance the numerical error in the prior approximation with the statistical errors in regression and classification tasks.

\section{Main Results}\label{sec:main}
In this section we obtain sufficient scalings of $n_h$ with respect to $N$ using spectral error analysis for FE eigenvalue problems and our criterion outlined in Subsection \ref{ssec:bg pc}. We
assume throughout that $\mathcal{D}=(0,L_1)\times\cdots\times (0,L_D)$ is a hyperrectangle and that the Laplacian in \eqref{eq:spde global} is supplemented with Neumann boundary condition, so that we have the following explicit expressions for its eigenvalues and eigenfunctions 
\begin{align}
    \Lambda_\i= \sum_{d=1}^D \frac{i_d\pi}{L_d},\quad \quad \Psi_\i(\x)=C_\i\prod_{d=1}^D \cos \left(\frac{i_d\pi x_d}{L_d}\right), \label{eq:Dd eigenpairs}
\end{align}
where $\i=(i_1,\ldots,i_D) \in \mathbb{N}^D$ is  a multi-index and $C_\i$'s are constants so that the $\Psi_\i$'s are $L^2(\mathcal{D})$-normalized. The Mat\'ern-type GP \eqref{eq:spde global} can then be written as
\begin{align}
    u=\kappa^{s-D/2}\sum_{\i\in\mathbb{N}^D} (\kappa^2+\Lambda_\i)^{-s/2} \xi_\i \Psi_\i,\quad \quad \xi_\i\overset{i.i.d.}{\sim} \mathcal{N}(0,1). \label{eq:approximate field rectangle}
\end{align}
The explicit expressions for the eigenpairs in \eqref{eq:Dd eigenpairs} allow us to establish the   following result \cite[Theorem 2.1]{khristenko2019analysis},   which shows that the covariance function of \eqref{eq:approximate field rectangle} is nearly indistinguishable from the classical Mat\'ern covariance function \eqref{eq:matern covariance}  away from the boundary. 
\begin{proposition}\label{prop:covariances are close}
Let $c(\x,\x')$ denote the covariance function of the Mat\'ern-type GP \eqref{eq:approximate field rectangle} and let $\cmat(\x,\x')$ be the Mat\'ern covariance function \eqref{eq:matern covariance} with 
\begin{align}
    \sigma^2 = \frac{\Gamma(s-\frac{D}{2})}{(4\pi)^{D/2}\Gamma(s)}. \label{eq:varaince formula}
\end{align}
Then 
\begin{align}
    c(\x, \x')=\sum_{\k\in\mathbb{Z}^D}\sum_{\mathbf{T}\in\mathcal{T}}\cmat( \mathbf{T}\x, \x'-2\k\mathbf{L}),\quad\quad \x, \x'\in (0,L_1)\times \cdots \times (0,L_D),\label{eq:Matern representation of rectangle field}
\end{align}
where $\mathcal{T}$ is the collection of all $D\times D$ diagonal matrices whose diagonal entries are either $1$ or $-1$, $\mathbf{T}\x$ denotes matrix-vector multiplication and $\k\mathbf{L}$ denotes $(k_1L_1,\ldots,k_DL_D)$.
\end{proposition}

Note that if the correlation range $\rho=\sqrt{8\nu}/\kappa$ (where, recall, $\nu = s - D/2)$ is much smaller than $\operatorname{min}_{d} L_{d},$ and in addition $\x, \x'$ are at a distance larger than $2\rho$ from each side of the hyperrectangle,  the only significant term that remains in \eqref{eq:Matern representation of rectangle field}  is $\cmat(\x,\x')$. Therefore, \eqref{eq:approximate field rectangle} gives a good approximation of the classical Mat\'ern model away from the boundary. In practice one can choose a larger hyperrectangle than the domain of interest to reduce the boundary effect \cite{lindgren2011explicit}, see also \cite{khristenko2019analysis}. Our focus on hyperrectangles also facilitates the concrete FE construction and error analysis in the next subsection. 

\subsection{FEM Construction and Spectral Error Bounds}\label{sec:fem hyper rectangle}
We shall construct the FE space on $[0,L_1]\times\cdots \times [0,L_D]$ as the tensor product of FE spaces on each interval $[0,L_d]$. To begin with, let $P$ be a uniform partition of $[0,L_d]$ into $K+1$ points with width $h=L/K$ and let $V_{h}$ be the space of continuous piecewise linear functions with respect to $P$. To simplify the notation we drop the dependence on $d$ below. Precisely, a basis of $V_h$ consists of 
\begin{align*}
    e_{h,k}=\begin{cases}
    h^{-1}x-k+1 \quad & x\in[(k-1)h,kh]\\
    -h^{-1}x+k+1  \quad & x\in[kh,(k+1)h]\\
    0 \quad & \operatorname{otherwise}
    \end{cases},
    \quad k=1,\ldots,K-1,
\end{align*}
with $e_{h,0}=(-h^{-1}x+1)\mathbf{1}_{[0,h]}$ and $e_{K,h}=(h^{-1}x-K+1)\mathbf{1}_{[(K-1)h,K h]}$. Let $\mathcal{J}_h$ be the Galerkin discretization of $\smash{\kappa^2-\frac{d^2}{dx^2}}$ over $V_h$. The eigenvalues $\{\lambda_{h,i}\}_{i=0}^K$ and eigenfunctions $\{\psi_{h,i}\}_{i=0}^K$ of $\mathcal{J}_h$ can be found by solving the generalized eigenvalue problem 
\begin{align*}
    \G \z=\lambda \M \z,
\end{align*}
where $\z$ represents the coordinates of $\psi_{h,i}$ in terms of the $e_{h,k}$'s and $\G, \M\in\mathbb{R}^{(K+1)\times (K+1)}$ are matrices with entries
\begin{align*}
    \G_{ij}=\frac{1}{h} \cdot\begin{cases}
    2\quad & i=j \notin \{1,K+1\}\\
    1\quad & i=j \in \{1,K+1\}\\
    -1\quad & |i-j|=1 \\
    0 \quad &\operatorname{otherwise}
    \end{cases},
    \quad \quad 
    \M_{ij}=h \cdot \begin{cases}
    2/3\quad & i=j \notin \{1,K+1\}\\
    1/3\quad & i=j \in \{1,K+1\}\\
    1/6\quad & |i-j|=1 \\
    0 \quad &\operatorname{otherwise}
    \end{cases}.
\end{align*}
One can check that 
\begin{align}
    \lambda_{h,i}=\frac{6}{h^2}\frac{1-\cos\left(i\pi h/L\right)}{2+\cos\left(i\pi h/L\right)}, \quad \psi_{h,i}=c_i\sum_{k=0}^{K}\cos\left(\frac{ki\pi h}{L}\right)e_{h,k}  ,\quad\quad   i,k\in[K], \label{eq:1d fem eigenpairs}
\end{align}
where   $c_i$'s are normalizing constants so that $\psi_{h,i}$ has $L^2(\mathcal{D})$ norm one, and $[K] = \{ 0, \ldots, K \}.$  We then have the following error estimates: 
\begin{lemma}\label{lemma:spectral bound 1d}
Let $\{(\lambda_i,\psi_i)\}_{i=1}^{\infty}$ be the eigenvalues and $L^2(\mathcal{D})$-orthonormal eigenfunctions of $\kappa^2-\frac{d^2}{dx^2}$ over $(0,L)$ with Neumann boundary condition. There is a constant $C$ so that, for $i\in[K]$,  
\begin{align*}
    |\lambda_{h,i}-\lambda_i|\leq C\lambda_i^2h^2, \quad \quad \|\psi_{h,i}-\psi_i\|_{\infty}\leq C\lambda_i h^2.
\end{align*}
Furthermore the $\psi_{h,i}$'s are also $L^2(\mathcal{D})$-orthonormal.
\end{lemma}
  
\begin{remark}
Eigenvalue estimates and eigenfunction estimates in $L^2$ norm can be found for instance in  \cite[Theorems 6.1 and 6.2]{strang1973analysis}, where more general elliptic operators and domains are considered. However, for our subsequent developments we need eigenfunction estimates in $L^{\infty}$ norm, and for this reason we include an elementary proof of Lemma \ref{lemma:spectral bound 1d} in Section  \ref{sec:proofs}. \qed
\end{remark} 

For Galerkin discretization of  $\kappa^2-\Delta$ on $[0,L_1]\times \cdots \times [0,L_D]$, let $\mathcal{P}$ be the uniform grid constructed by uniformly partitioning each interval with $K_d+1$ nodes so that $h_d=L_d/K_d$ in each dimension. Define for $\mathbf{h}=(h_1,\ldots,h_D)$ the FE space  
\begin{align*}
    \mathcal{V}_{\h}=V_{h_1} \otimes \cdots\otimes V_{h_d}:=
    \left\{v(\x)=\prod_{d=1}^D v_{h_d}(x_d): v_{h_d}\in V_{h_d} \right\},
\end{align*}
where $V_{h_d}$ is the FE space on $[0,L_d]$ constructed above. It can be shown that the eigenvalues $\Lambda_{\h,\i}$ and  eigenfunctions $\Psi_{\h,\i}$ of $\mathcal{L}_{\h}$ (the Galerkin discretization of $\kappa^2-\Delta$)  are
\begin{align*}
    \Lambda_{\h,\i}=\sum_{d=1}^D \lambda_{h_d,i_d},\quad \Psi_{\h,\i}(\x)=\prod_{d=1}^D \psi_{h_d,i_d}(x_d), \quad \quad \i\in [K_1]\times\cdots\times[K_D],
\end{align*}
where the $\lambda_{h_d,i_d}$'s and $\psi_{h_d,i_d}$'s are as in \eqref{eq:1d fem eigenpairs}. Indeed for $v_{\h}(\x)=\prod_{d=1}^D v_{h_d}(x_d)\in\mathcal{V}_{\h}$ we have that 
\begin{align*}
    \langle \nabla \Psi_{\h,\i}, \nabla v_{\h}\rangle &= \int_{\mathcal{D}}  \sum_{d=1}^D \Big(\psi^{\prime}_{h_d,i_d} v^{\prime}_{h_d} \prod_{\ell\neq d}  \psi_{h_{\ell},i_{\ell}} v_{h_{\ell}} \Big)d\x\\
    &= \sum_{d=1}^D \langle \psi^{\prime}_{h_d,i_d}, v^{\prime}_{h_d}\rangle \prod_{ \ell\neq d} \langle \psi_{h_{\ell},i_{\ell}},v_{h_{\ell}}\rangle =\sum_{d=1}^D \lambda_{h_d,i_d} \prod_{\ell=1}^D \langle \psi_{h_{\ell},i_{\ell}}, v_{h_{\ell}} \rangle=\sum_{d=1}^D \lambda_{h_d,i_d} \langle \Psi_{\h,\i},v_{\h}\rangle,
\end{align*}
where the primes denote weak derivatives. Moreover the $\Psi_{\h,\i}$'s are orthonormal since the $\psi_{h,i}$'s are and hence they form a complete set of eigenbasis for $\mathcal{L}_{\h}$. 
The following error estimates are immediate, where we recall that the true eigenpairs are given in \eqref{eq:Dd eigenpairs}: 
\begin{lemma}\label{lemma:rectangle FEM eigenfucntion infty bound}
 For $\i\in[K_1]\times\cdots\times[K_D]$ we have 
\begin{align*}
    |\Lambda_{\h,\i}-\Lambda_\i|\leq C\Lambda_\i^2h^2,\quad \quad \|\Psi_{\h,\i}-\Psi_\i\|_{\infty} \leq C\Lambda_\i h^2,
\end{align*}
where $h=\operatorname{max}_d h_d$ and $C$ is a constant depending only on $D$ and the $L_d$'s.
\end{lemma}
\begin{remark}
Since $\mathcal{D}$ is a bounded domain, we obtain also the $L^2(\mathcal{D})$ bound $\|\Psi_{\h,\i}-\Psi_\i\|_2 \leq C\Lambda_\i h^2$.  \qed
\end{remark}
Since the approximation error in Lemma \ref{lemma:rectangle FEM eigenfucntion infty bound} depends on $h=\operatorname{max}_d h_d$, we shall from now on assume that the $h_d$'s are chosen so that they are of the same order, i.e., $\operatorname{max}_{j\neq k}\frac{h_j}{h_k}=O(1)$ as $h\rightarrow 0$, and treat only $h$ as the mesh size. As a consequence the total number of grid points satisfies the following scaling
\begin{align}
    n_{\h}=\prod_{d=1}^D (L_d/h_d+1)\asymp h^{-D}. \label{eq:total number of grid points}
\end{align}

\subsection{Balancing Numerical and Statistical Errors}
Now we use the spectral error bounds in Lemma \ref{lemma:rectangle FEM eigenfucntion infty bound} to obtain $L^2(\mathcal{D})$ and $L^\infty(\mathcal{D})$ error bounds for FE representations of Mat\'ern-type GP priors (Lemma \ref{lemma:infty field rate rectangle}). These prior bounds, combined with Proposition \ref{prop:prior approx implies same pc rate}, will yield our main result (Theorem \ref{thm:rectangle}). Let 
\begin{align*}
    u_{\h}=\kappa^{s-D/2}\sum_{\i\in[K_1]\times\cdots\times[K_D]}(\kappa^2+\Lambda_{\h,\i})^{-s/2}\xi_{\i} \Psi_{\h,\i},\quad \quad \xi_{\i}\overset{i.i.d.}{\sim} \mathcal{N}(0,1),
\end{align*}
be the FE representation of the Mat\'ern-type GP $u$  in \eqref{eq:approximate field rectangle}. Recall that we are interested in estimating the function $f_0(\x)=\mathbb{E}[Y|\X=\x]$ based on i.i.d. samples $\{(\X_i,Y_i)\}_{i=1}^N$ with prior $\Pi_N=\operatorname{Law}\bigl(u_{\h_N}\bigr)$ for the fixed design regression setting and $\Pi_N=\operatorname{Law}\bigl(\Phi(u_{\h_N})\bigr)$ for the binary classification setting, where $\h_N=(h_{N,1}, \ldots,h_{N,D})$ is to be determined. Based on the discussion in Subsection \ref{ssec:bg pc}, it suffices to quantify the approximation error of $u$ defined in \eqref{eq:approximate field rectangle} by $u_{\h_N}$.

\begin{lemma}\label{lemma:infty field rate rectangle}
Recall that $h=\operatorname{max}_d h_d$. Suppose $s>D/2$. It holds that
\begin{align*}
    \mathbb{E}\|u_{\h}-u\|_2^2\leq C\kappa^{2s-D}h^{(2s-D)\wedge 4},
\end{align*}
where $C$ is a constant independent of $\kappa$ and $h$. 
Furthermore the Mat\'ern-type GP $u$ defined in \eqref{eq:approximate field rectangle} belongs almost surely to $\mathcal{C}^{\beta}(\mathcal{D})$ for $0<\beta<1\wedge (s-D/2)$. Moreover, for $s>D$  it holds that
\begin{align*}
    \mathbb{E}\|u_{\h}-u\|_{\infty}^2\leq C \kappa^{2s-D} h^{(2s-2D)\wedge 4},
\end{align*}
where $C$ is a constant independent of $\kappa$ and $h$. 
\end{lemma}

\begin{remark}
The $L^2$ error bound has been shown to hold in greater generality, see e.g. \cite[Theorem 2.10]{bolin2020numerical} and \cite[Theorem 2]{cox2020regularity}. A sharper $L^{\infty}$ error bound was shown in \cite[Theorem 3]{cox2020regularity} when $D=1$, while our result holds for general dimension $D$. 
\end{remark}

As a corollary of Proposition \ref{prop:prior approx implies same pc rate} we have the following main result, presented in terms of the scaling of $h_N=\operatorname{max}_d h_{N,d}$. Notice that the concentration function defined in \eqref{eq:concentration function} depends implicitly on $s$ through $u$.

\begin{theorem}\label{thm:rectangle}
\begin{enumerate}
    \item Fixed design regression: Consider the Mat\'ern-type GP $u$ defined by \eqref{eq:KL expansion true field} with $s>D$. Suppose $\varepsilon_N$ satisfies $\phi_{f_0}(\varepsilon_N;u, \|\cdot\|_{\infty})\leq N\varepsilon_N^2$.  Set 
\begin{align}
    h_{N}\asymp 
    N^{-\frac{1}{(2s-2D)\wedge 4}} 
     \label{eq:scaling of hn rectangle L infty}
\end{align} 
with a large enough proportion constant. 
Then, for any sufficiently large $M>0,$
\begin{align*}
    \mathbb{E}_{f_0} \Pi_N\left(f:\|f-f_0\|_N\leq M\varepsilon_N\,|\,\{(\X_i,Y_i)\}_{i=1}^N\right)\xrightarrow{N\rightarrow\infty} 1,
\end{align*}
where we recall $\|f\|_N^2=N^{-1}\sum_{i=1}^N |f(\X_i)|^2$. 

\item Binary classification: Consider the Mat\'ern-type GP $u$ defined by \eqref{eq:KL expansion true field} with $s>D/2$. Suppose $\varepsilon_N$ satisfies $\phi_{\Phi^{-1}(f_0)}(\varepsilon_N; u,\|\cdot\|_2)\leq N\varepsilon_N^2$.  Set 
\begin{align}
    h_{N}\asymp 
    N^{-\frac{1}{(2s-D)\wedge 4}}
     \label{eq:scaling of hn rectangle L2}
\end{align} 
with a large enough proportion constant. 
Then, for any sufficiently large $M>0,$
\begin{align*}
    \mathbb{E}_{f_0} \Pi_N\left(f:\|f-f_0\|_2\leq M\varepsilon_N\,|\,\{(\X_i,Y_i)\}_{i=1}^N\right)\xrightarrow{N\rightarrow\infty} 1.
\end{align*}
\end{enumerate}
\end{theorem}

\begin{remark}\label{remark:rate pm}
  
Theorem \ref{thm:rectangle} provides a scaling of $h_{N}$ so that the sequence of posteriors with respect to the FE prior $\Pi_N$ achieves the same contraction rate as if the Mat\'ern-type prior $\Pi$ was used. We remark that a refined analysis of the rate at which the posterior probabilities go to $1$ could be used to obtain similar conclusions for the posterior means under suitable assumptions, i.e. 
\begin{align*}
    \mathbb{E}_{f_0} d_N(\widehat{f},f_0)^2 &\lesssim \varepsilon_N^2, \qquad \widehat{f}=\int f\,d\Pi\big(f|\{(\X_i,Y_i)\}_{i=1}^N\big),\\
    \mathbb{E}_{f_0} d_N(\widehat{f}_N,f_0)^2 &\lesssim \varepsilon_N^2, \qquad \widehat{f}_N=\int f\,d\Pi_N\big(f|\{(\X_i,Y_i)\}_{i=1}^N\big).
\end{align*}
In other words, the sequence of posterior means with respect to $\Pi_N$ converges to $f_0$ at the same rate as those with respect to $\Pi$, thereby giving a more interpretable conclusion. For fixed design regression, this follows from \cite[Theorem 1]{van2011information} and Jensen's inequality with $d_N=\|\cdot\|_{N}$. For binary classification, using again Jensen's inequality and the fact that $|f|\leq 1$ we have  
\begin{align*}
    \|\widehat{f}-f_0\|_2^2 &\leq \int \|f-f_0\|_2^2\, d\Pi\big(f|\{(\X_i,Y_i)\}_{i=1}^N\big) \\
    &\leq M^2\varepsilon_N^2 + 4|\mathcal{D}| \Pi\big(f:\|f-f_0\|_2\geq M\varepsilon_N |\{(\X_i,Y_i)\}_{i=1}^N\big),
\end{align*}
where $|\mathcal{D}|$ is the Lebesgue measure of $\mathcal{D}$. Therefore a rate faster than $\varepsilon_N^2$ on the decay of the posterior probability suffices, which is satisfied under mild assumptions \cite[Theorems 2.2 and 2.3]{ghosal2000convergence}. \qed 
\end{remark}

\begin{remark}\label{remark:proportion constant}
For the regression setting, \eqref{eq:scaling of hn rectangle L infty} together with \eqref{eq:total number of grid points} gives the scaling for the total number of grid points needed,
\begin{align*}
    n_{\h_N} \asymp N^{\frac{D}{(2s-2D)\wedge 4}}.
\end{align*}
In particular when $s>3D/2$, $D=1,2,3$, the exponent for $N$ is less than one and we have $n_{\h_N}\ll N$ asymptotically. For classification, $s> D$ suffices. 
However, we remark that the proportion constant depends implicitly on $\kappa$ and the $L_d$'s as can be seen from \eqref{eq:total number of grid points} and Lemma \ref{lemma:infty field rate rectangle}. In particular, if both $\kappa$ and the $L_d$'s are large, which reflects the case of a rapidly changing field over a large spatial domain, then $N$ may need to be large enough in order for $n_{\h_N}$ to be smaller than $N$. We shall demonstrate through simulation studies in Section \ref{sec:numerics} that for moderate $\kappa$ and $L_d$'s one can achieve $n_{\h}< N$ when $N=O(10^2)$ for a one-dimensional example and $N=O(10^3)$ for a two-dimensional one, thereby suggesting that Theorem \ref{thm:rectangle} has some practical implication. 
\qed
\end{remark}

The scaling of $h_N$ in Theorem \ref{thm:rectangle} ensures that the numerical errors in the FE representations of a true Mat\'ern-type GPs $u$ do not impact the corresponding contraction rates. In the remainder of this section we give an example where the rates $\varepsilon_N$ with respect to the true Mat\'ern-type GP $u$ can be explicitly computed under a smoothness assumption on the truth $f_0$. For this purpose we introduce a notion of regularity of $f_0$ based on the orthonormal basis $\{\Psi_{\i}\}_{\i\in\mathbb{N}^D}$. Let $S$ be an even function in the Schwartz space $\mathcal{S}(\mathbb{R})$ satisfying 
\begin{align*}
    0\leq S \leq 1,\quad S \equiv 1\,\, \operatorname{on}\,\, \Bigl[-\frac12,\frac12\Bigr], \quad \operatorname{supp}(S)\subset [-1,1].
\end{align*}
Define the space 
\begin{align*}
    B_{\infty,\infty}^{\beta}=\left\{f=\sum_{\i\in\mathbb{N}^D}f_{\i}\Psi_{\i}: \|f\|_{B_{\infty,\infty}^{\beta}}=\underset{j\in\mathbb{N}}{\operatorname{sup}}\,2^{\beta j} \|S_j(\sqrt{\Delta})f(\cdot)-f(\cdot)\|_{\infty}<\infty\right\},
\end{align*}
where $S_j(\cdot)=S(2^{-j}\cdot)$ and
\begin{align*}
    S_j(\sqrt{\Delta}) f = \sum_{\i\in\mathbb{N}^D} S_j(\sqrt{\Lambda_{\i}}) f_{\i} \Psi_{\i}.
\end{align*}
\begin{proposition}\label{lemma:besov example infnity rectangle}
Suppose $f_0\in B_{\infty,\infty}^{\beta}$ and set $s=\beta+\frac{D}{2}$ in the definition of $u$. Then for $\varepsilon_N$ a large enough multiple of $N^{-\beta/(2\beta+D)}$, we have $\phi_{f_0}(\varepsilon_N; u, \|\cdot\|_2)\leq N\varepsilon_N^2$ and $\phi_{f_0}(\varepsilon_N; u, \|\cdot\|_{\infty})\leq N\varepsilon_N^2$.
\end{proposition}
The space $B_{\infty,\infty}^{\beta}$ can be seen as a Besov-type space tailored to our specific setting, where the prior support associated with the Mat\'ern-type GP $u$ consists of functions defined as series expansions in terms of the $\Psi_{\i}$'s. Similar function spaces have been considered in \cite{castillo2014thomas}. As the usual Besov spaces, functions in $B_{\infty,\infty}^{\beta}$ should be understood to have regularity of order $\beta$, in which case the contraction rate $\smash{N^{-\beta/(2\beta+D)}}$ matches the usual minimax optimal rate for estimating $\beta$-regular functions.

\section{Simulation Study}\label{sec:numerics}
The aim of this section is to complement the understanding given by Theorem \ref{thm:rectangle} through numerical simulations in the regression setting. We consider one and two-dimensional examples in Subsections \ref{sec:ex1} and \ref{sec:ex2},  respectively.

The general set up is as follows. Let $\{\X_i\}_{i=1}^N$ be fixed design points in the domain $\mathcal{D}$ and $\{Y_i\}_{i=1}^N$ be noisy observations generated from 
\begin{align*}
    Y_i=f_0(\X_i)+\eta_i ,\quad \quad \eta_i\overset{i.i.d.}{\sim}\mathcal{N}(0,\tau^2),
\end{align*}
where $f_0$ is the ground truth and $\tau$ is known. We compare two approaches for inferring $f_0$, namely the covariance function (CF) approach and the finite element (FE) approach with mass lumping. They can be summarized as follows:
\begin{align*}
    \y\sim \mathcal{N}(\mathbf{f_N},\tau^2 I_N), \quad \mathbf{f_N} \sim \mathcal{N}(\0,\mathbf{\Sigma}) \quad \Longrightarrow \quad \widehat{\mathbf{f}}_{\mbox {\tiny{\rm CF}}}&= \mathbf{\Sigma}(\mathbf{\Sigma}+\tau^2\mathbf{I_N})^{-1}\y,
\end{align*}
where $\mathbf{\Sigma}=\{\cmat(\X_i,\X_j)\};$ and 
\begin{align*}
    \y\sim \mathcal{N}(\S\w,\tau^2I_N),\quad 
    \w \sim \mathcal{N}(\0,\Q^{-1})\quad \Longrightarrow\quad  \widehat{\mathbf{f}}_{\mbox {\tiny{\rm FE}}}= \S (\S^\top \S+\tau^2\Q)^{-1}\S^\top \y,
\end{align*}
where $\S_{ij}=e_j(X_i)$ is as in Subsection \ref{ssec: FEM for regression and classification} and $\Q=(\kappa^2\M+\G)\big[\widetilde{\M}^{-1}(\kappa^2 \M+\G)\big]^{s-1}$ as in Proposition \ref{prop:covariance of FEM approximation} but with the lumped mass matrix $\smash{\widetilde{\M}}$ instead. As noted in Remark \ref{remark:rate pm}, we shall compare the error $\|\widehat{\mathbf{f}}_{\mbox {\tiny{\rm CF}}}-\mathbf{f}_0\|_N$ and $\|\widehat{\mathbf{f}}_{\mbox {\tiny{\rm FE}}}-\mathbf{f}_0\|_N$ when an increasing number of grid points ($n_h$) is used in the FE approach, where $\mathbf{f}_0=(f_0(\X_1),\ldots,f_0(\X_N))^\top$ and $\|\cdot\|_N$ is the vector 2-norm normalized by $1/\sqrt{N}$. Note that the CF and FE approaches studied here are not exactly those analyzed in Theorem \ref{thm:rectangle}, i.e., the error of going from the CF approach to the Mat\'ern-type prior (expected to be small by Proposition \ref{prop:covariances are close}) and that of the lumped mass procedure were not accounted for. 
However, we remark that both errors do not lead to a significant difference in the numerical results and we will only focus on the CF and FE approaches, which are used in practice.

\begin{figure}[!htb]
\minipage{1\textwidth}
\centering
\minipage{0.333\textwidth}
  \includegraphics[width=\linewidth]{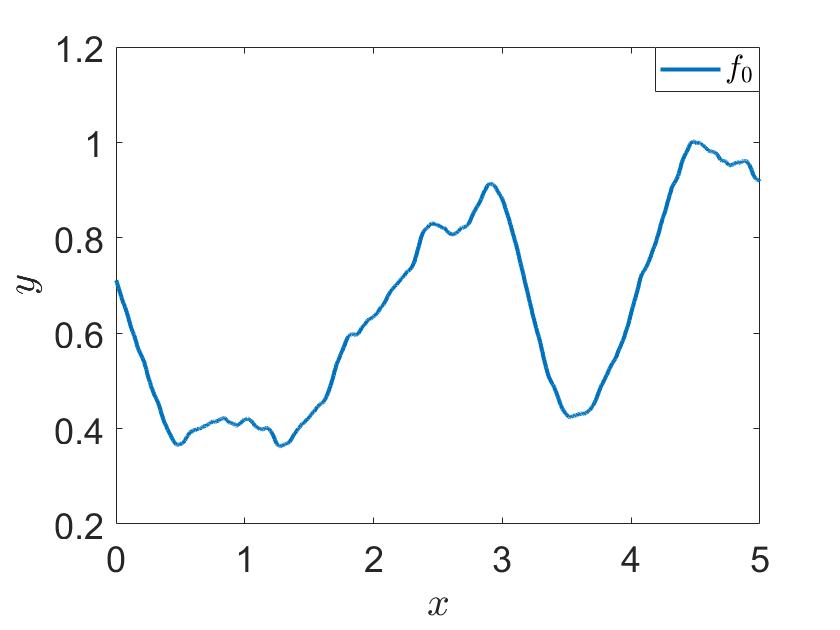}
\vspace{-10pt}
\endminipage
\minipage{0.333\textwidth}
  \includegraphics[width=\linewidth]{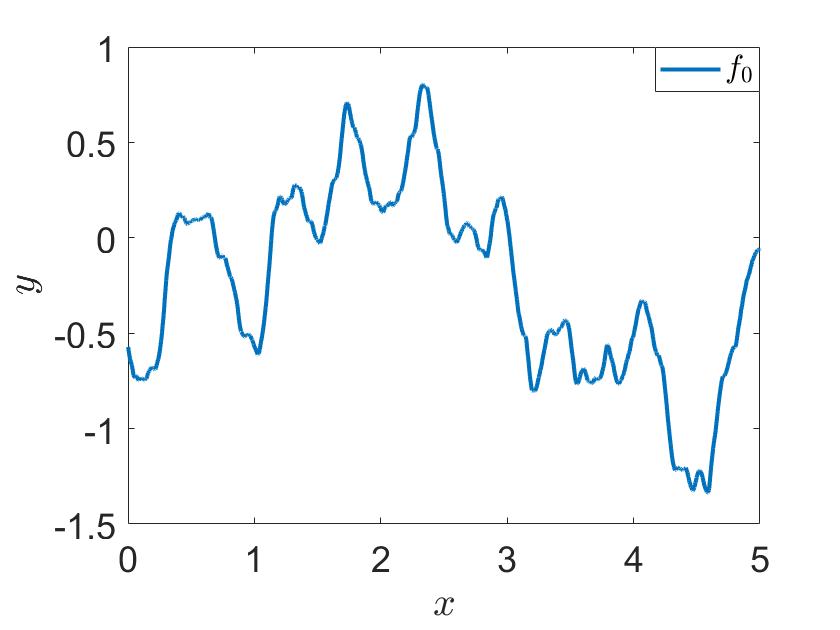}
\vspace{-10pt}
\endminipage 
\minipage{0.333\textwidth}
  \includegraphics[width=\linewidth]{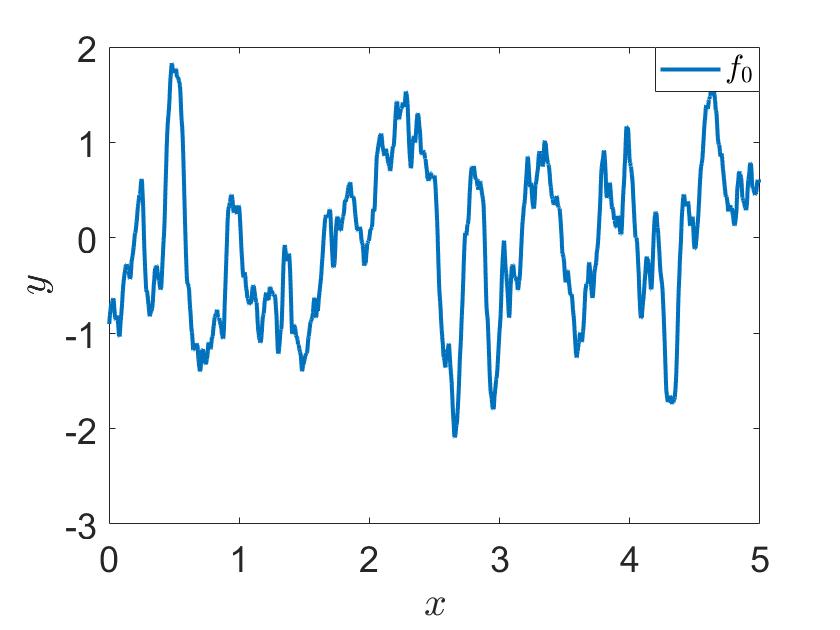}
\vspace{-10pt}
\endminipage 
\endminipage\hfill

\minipage{1\textwidth}
\centering
\minipage{0.333\textwidth}
  \includegraphics[width=\linewidth]{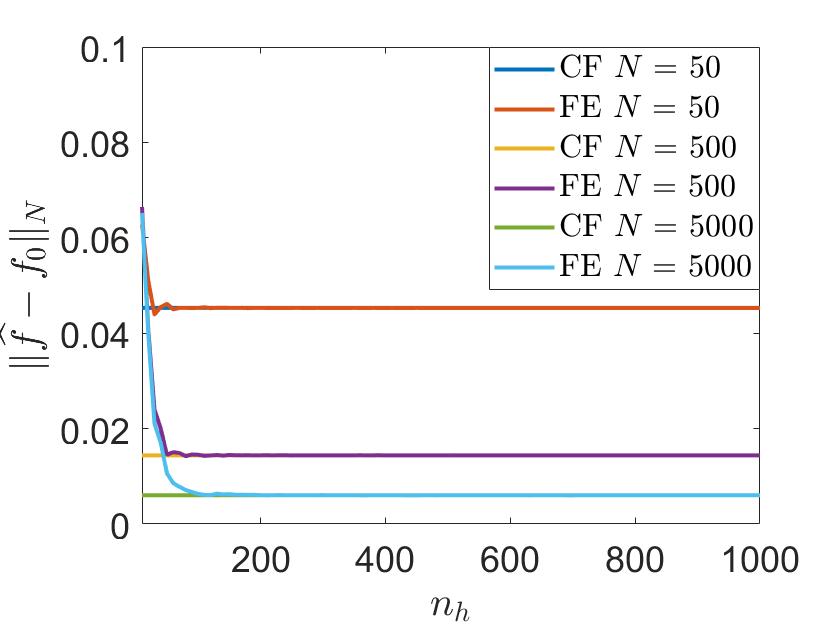}
\vspace{-10pt}\subcaption{$s_0=2$, $\kappa=1$}
\endminipage
\minipage{0.333\textwidth}
  \includegraphics[width=\linewidth]{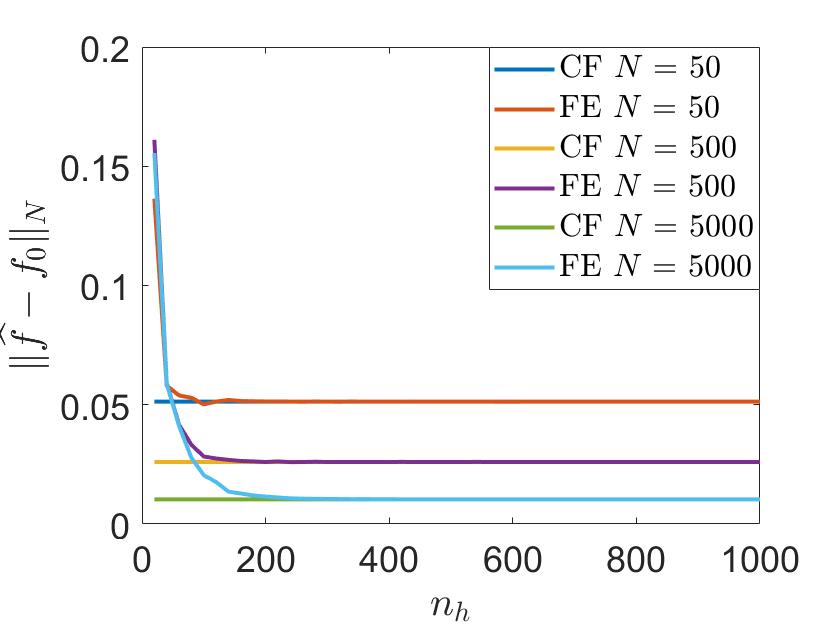}
\vspace{-10pt}\subcaption{$s_0=2$, $\kappa=5$}
\endminipage 
\minipage{0.333\textwidth}
  \includegraphics[width=\linewidth]{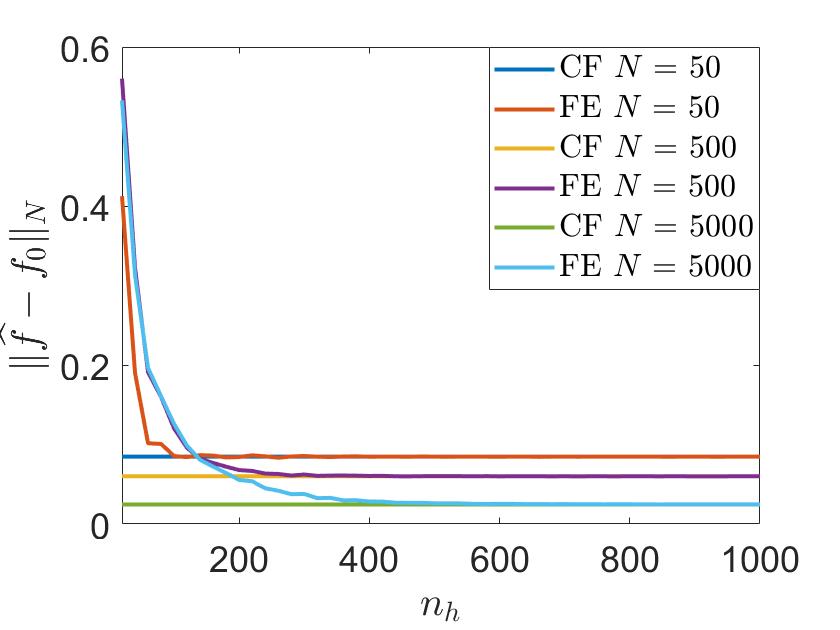}
\vspace{-10pt}\subcaption{$s_0=2$, $\kappa=25$}
\endminipage 
\endminipage\hfill
\caption{The three columns represent simulations for $\kappa_0=1$, 5 and 25 respectively with $s_0=2$ in all cases. The upper row shows plots of $f_0$. The lower row compares the estimation error $\|\widehat{\mathbf{f}}-\mathbf{f}_0\|_N$ between the covariance function (CF) approach and the finite element (FE) approach as $n_h$ increases, for three levels of data $N=50$, 500 and 5000. } \label{figure:ex1}
\end{figure}

\subsection{The One-dimensional Case}\label{sec:ex1}
To start with, let $\{\X_i\}_{i=1}^N$ be fixed design points generated from the uniform distribution over $[0,L]$. 
We shall consider $f_0$'s generated from the following series expansion (with a sufficiently high truncation)
\begin{align}
    f_0(x)\sim \frac{\kappa_0^{-1/2}\xi_0}{\sqrt{L}}+\frac{\sqrt{2}\kappa_0^{s_0-1/2}}{\sqrt{L}}\sum_{i=1}^{\infty}\left[\kappa_0^2+\left(\frac{i\pi }{L}\right)^2\right]^{-s_0/2}\left[\xi_i\cos\left(\frac{i\pi x}{L}\right)+\zeta_i\sin\left(\frac{i\pi x}{L}\right)\right], \label{eq:truth distribution}
\end{align}
where $\xi_i,\zeta_i\smash{\overset{i.i.d.}{\sim}}\mathcal{N}(0,1)$. Notice that \eqref{eq:truth distribution} is defined in the same spirit as \eqref{eq:KL expansion true field} except that the full trigonometric basis is used, so that the random field \eqref{eq:truth distribution} does not have a prescribed boundary condition. Our motivation to not consider here a Neumann boundary condition is to illustrate that similar conclusions as those suggested by our theory can be expected in more general settings.
Notice again that there are two parameters $s_0$ and $\kappa_0$, which control the smoothness and correlation lengthscale respectively. We will vary both $s_0$ and $\kappa_0$ in the following simulations. 

For both the CF and FE approaches, we use the same parameters $s_0$ and $\kappa_0$ that are used to generate $f_0$. In other words, we consider the Mat\'ern covariance \eqref{eq:matern covariance} with parameters $\nu=s_0-1/2$, $\kappa=\kappa_0$ and $\sigma^2$ given in \eqref{eq:varaince formula}, and FE approximation \eqref{eq:KL fem approx} with $s=s_0$ and $\kappa=\kappa_0$. For the FE approach, we construct the approximation over the larger interval $[-\rho, L+\rho]$ where $\rho=\sqrt{8\nu}/\kappa$ to reduce the boundary effects suggested in Proposition \ref{prop:covariances are close}. Three levels of data $N=50$, 500 and 5000 are considered and, for each $N$, we study the performance for the FE approach as the number $n_h$ of grid points increases. Finally we let $L=5$ and $\tau=0.1\cdot \|\mathbf{f}_0\|_2/\sqrt{N}$, which amounts to about 10\% error. 

Figure \ref{figure:ex1} shows the results when we fix the smoothness $s_0=2$ and vary $\kappa_0=1$, 5 and 25. We see that the estimation error for the FE approach decreases to that of the CF approach after certain threshold $n_h^{\ast}$. In other words, discretization at the level of $n_h^{\ast}$ for the FE approach is sufficient to yield the same estimation performance as the CF approach. The value of $n_h^{\ast}$ is seen to be smaller than the sample size when $N=500$ and is of an order of magnitude smaller when $N=5000$, in the same spirit as the scaling suggested in Theorem \ref{thm:rectangle}. The fact that $n_h^{\ast}$ is larger than the sample size when $N=50$ can be explained by the large proportion constant in Remark \ref{remark:proportion constant}. Furthermore such proportion constant increases with $\kappa$, as suggested by the larger $n_h^{\ast}$ for a larger $\kappa$.

To further understand the effect of the smoothness $s_0$, we perform two more simulations for (a) $s_0=1$, $\kappa_0=1$ and (b) $s_0=3$, $\kappa_0=25$. For (a) we see in Figure \ref{figure:ex1 2 1} that the $n_h^{\ast}$'s in this case are much larger than the $s_0=2$ cases. This is due to the roughness of the truth and the prior used and hence a large number of grid points are needed for accurate approximation even if $\kappa$ is small. On the other hand when $s_0=3$, Figure \ref{figure:ex1 2 2} shows qualitatively similar results as in Figure \ref{figure:ex1} in the sense that $n_h^{\ast}$ is asymptotically much smaller than $N$. Moreover the $n_h^{\ast}$'s are seen to be smaller than those when $s_0=2$, $\kappa_0=25$, as the underlying field is smoother and the required scaling suggested by Theorem \ref{thm:rectangle} is smaller. 

\begin{figure}[!htb]
\minipage{1\textwidth}
\centering
\minipage{0.333\textwidth}
  \includegraphics[width=\linewidth]{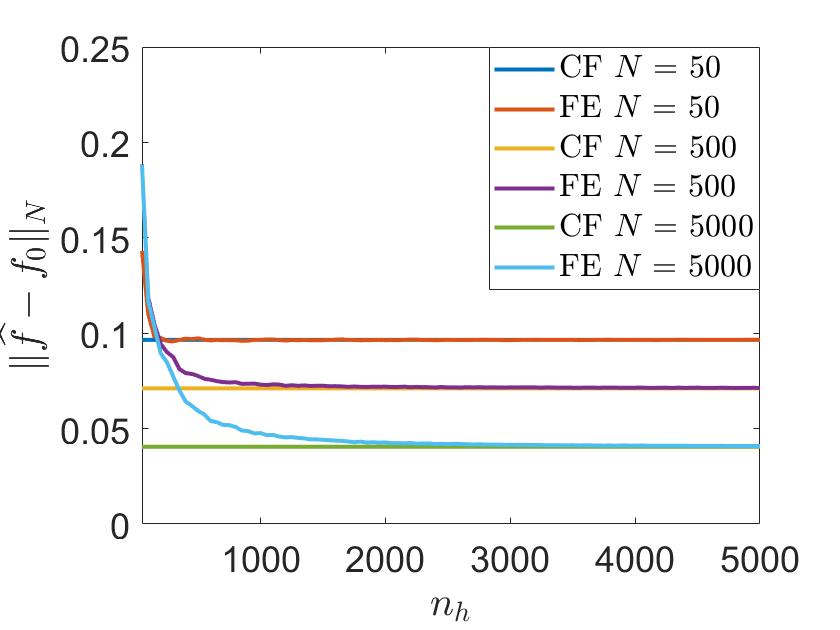}
\vspace{-10pt}\subcaption{$s_0=1$, $\kappa=1$.}\label{figure:ex1 2 1}
\endminipage 
\minipage{0.333\textwidth}
  \includegraphics[width=\linewidth]{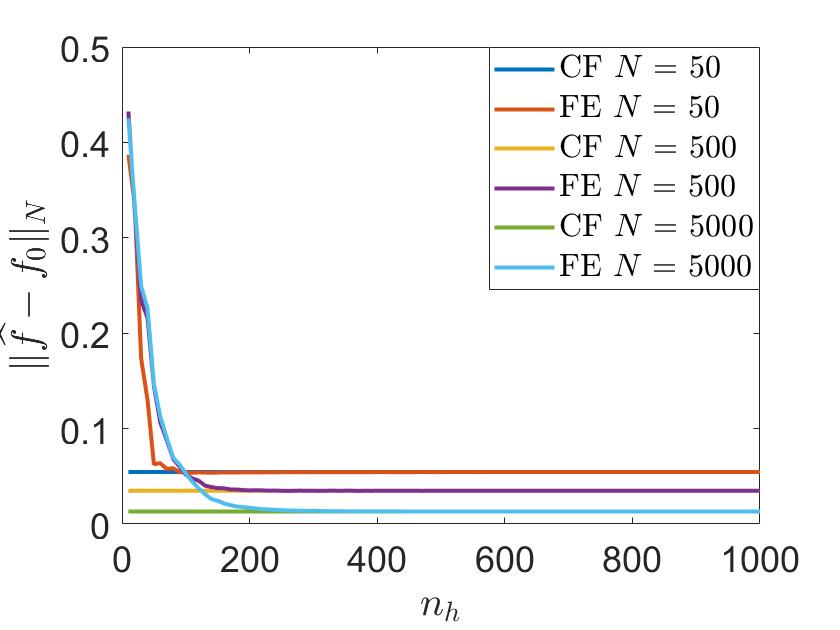}
\vspace{-10pt}\subcaption{$s_0=3$, $\kappa=25$.}\label{figure:ex1 2 2}
\endminipage 
\endminipage\hfill
\caption{Comparison of estimation error $\|\widehat{\mathbf{f}}-\mathbf{f}_0\|_N$ between the covariance function (CF) approach and the finite element (FE) approach on three data levels $N=50$, 500, 5000 for (a) $s_0=1$, $\kappa=1$ and (b) $s_0=3$, $\kappa=25$.}  \label{figure:ex1 2}
\end{figure}

\begin{figure}[!htb]
\minipage{1\textwidth}
\centering
\minipage{0.333\textwidth}
  \includegraphics[width=\linewidth]{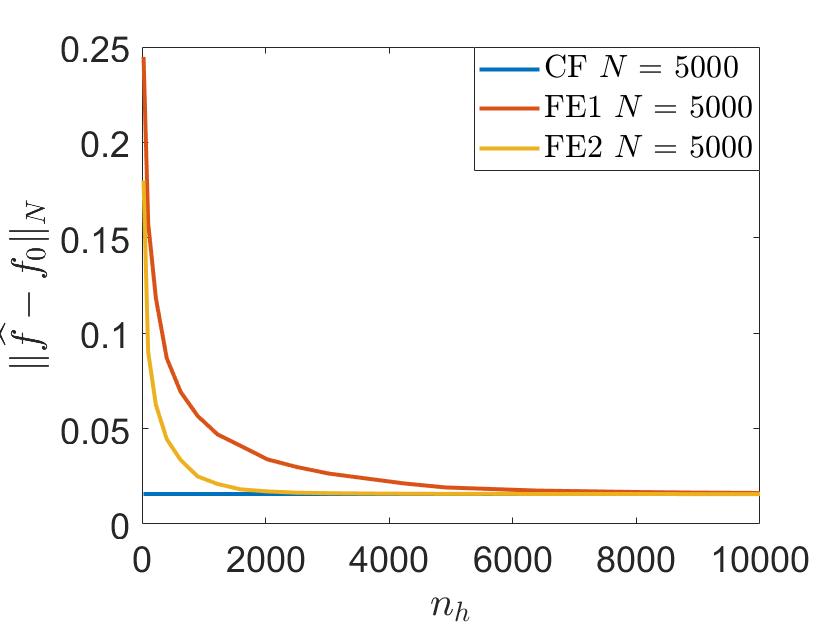}
\vspace{-10pt}\subcaption{$s_0=2$, $\kappa_0=1$.}
\endminipage
\minipage{0.333\textwidth}
  \includegraphics[width=\linewidth]{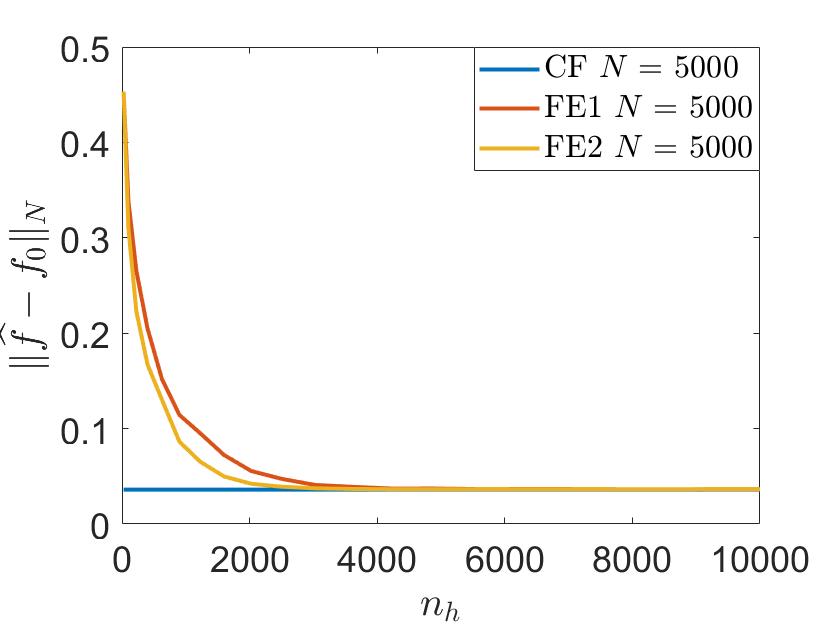}
\vspace{-10pt}\subcaption{$s_0=2$, $\kappa_0=5$.}
\endminipage 
\minipage{0.333\textwidth}
  \includegraphics[width=\linewidth]{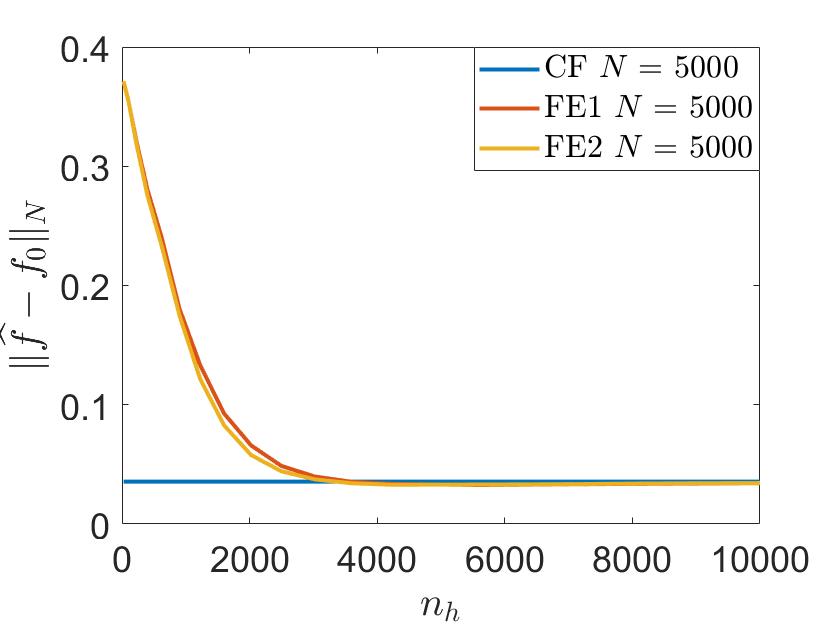}
\vspace{-10pt}\subcaption{$s_0=2$, $\kappa_0=25$.}
\endminipage 
\endminipage\hfill

\minipage{1\textwidth}
\centering
\minipage{0.333\textwidth}
  \includegraphics[width=\linewidth]{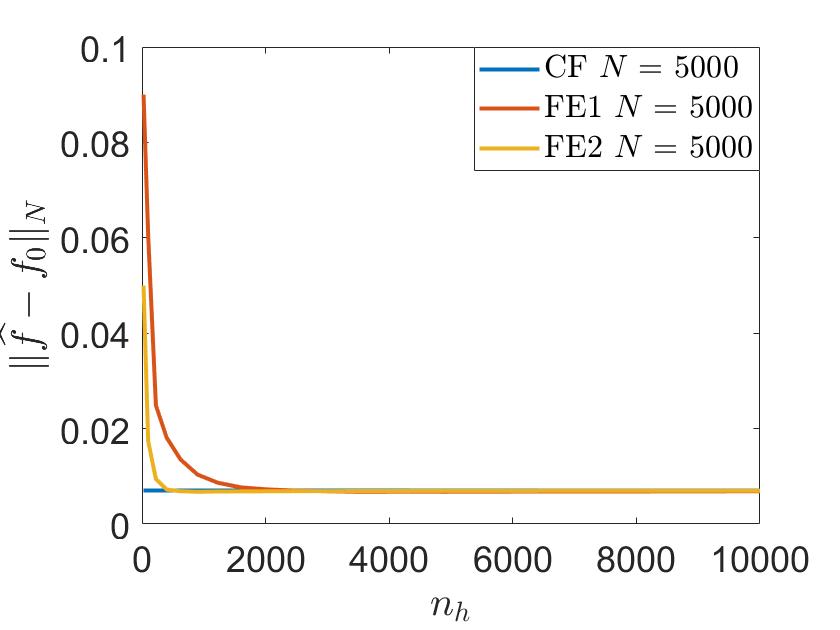}
\vspace{-10pt}\subcaption{$s_0=3$, $\kappa_0=1$.}
\endminipage
\minipage{0.333\textwidth}
  \includegraphics[width=\linewidth]{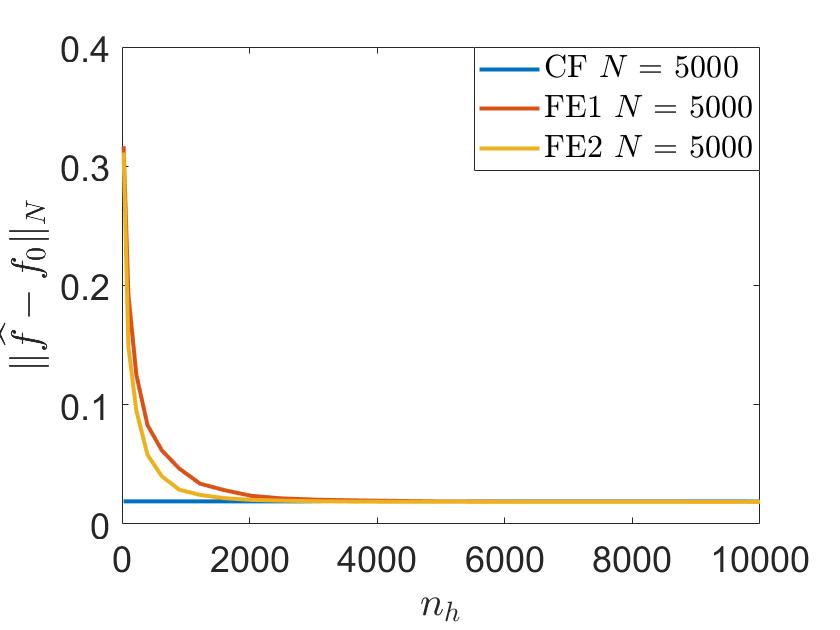}
\vspace{-10pt}\subcaption{$s_0=3$, $\kappa_0=5$.}
\endminipage 
\minipage{0.333\textwidth}
  \includegraphics[width=\linewidth]{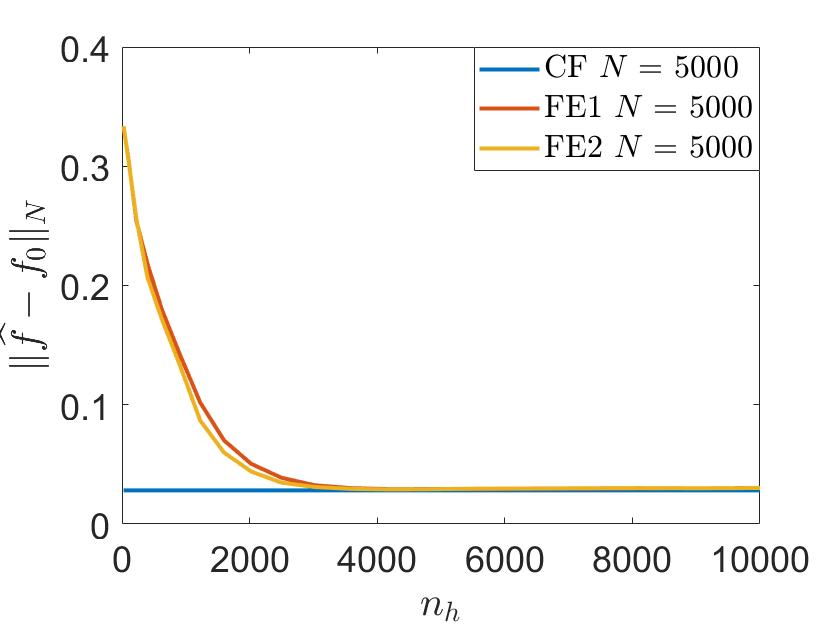}
\vspace{-10pt}\subcaption{$s_0=3$, $\kappa_0=25$.}
\endminipage 
\endminipage\hfill

\minipage{1\textwidth}
\centering
\minipage{0.333\textwidth}
  \includegraphics[width=\linewidth]{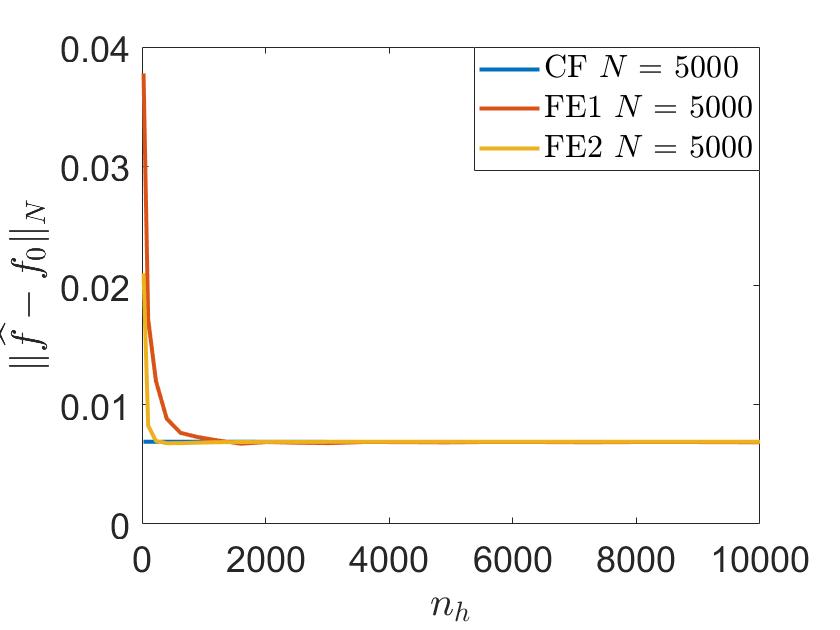}
\vspace{-10pt}\subcaption{$s_0=4$, $\kappa_0=1$.}
\endminipage
\minipage{0.333\textwidth}
  \includegraphics[width=\linewidth]{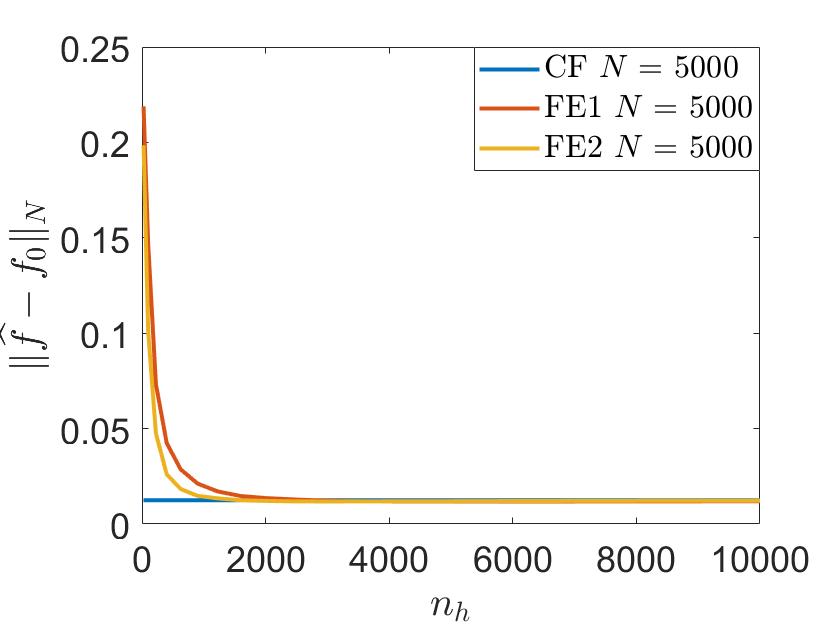}
\vspace{-10pt}\subcaption{$s_0=4$, $\kappa_0=5$.}
\endminipage 
\minipage{0.333\textwidth}
  \includegraphics[width=\linewidth]{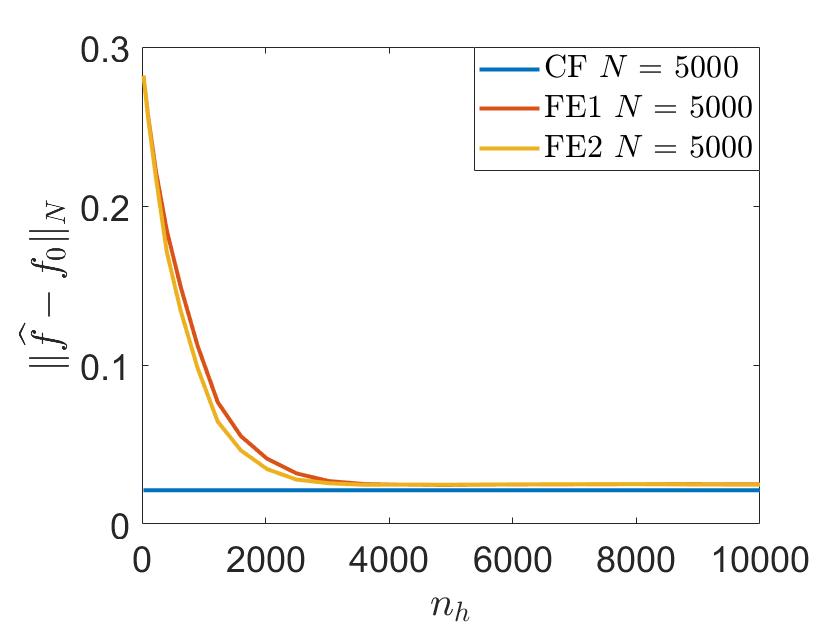}
\vspace{-10pt}\subcaption{$s_0=4$, $\kappa_0=25$.}
\endminipage 
\endminipage\hfill
\caption{Comparison of the estimation error $\|\widehat{\mathbf{f}}-\mathbf{f}_0\|_N$ as $n_h$ increases between the covariance function (CF) approach and two finite element approaches where FE1 is computed over $[-\rho, L+\rho]^2$ and FE2 is computed over $[-0.1\rho, L+0.1\rho]^2$. Simulation results for different combinations of $s_0$ and $\kappa_0$ are shown. } \label{figure:ex2}
\end{figure}

\subsection{The Two-dimensional Case}\label{sec:ex2}
Now we move on to the more practically relevant two-dimensional case following a similar set up as above. Let $\{\X_i\}_{i=1}^N$ be fixed design points generated from the uniform distribution over the square $[0,L]^2$ and $f_0$ be generated similarly as \eqref{eq:KL expansion true field} with $\psi_i$'s the full trigonometric basis (i.e. elements of the form $\sin(10\pi x_1/L)\sin(5\pi x_2/L)$, $\cos(3\pi x_1/L)\sin(9\pi x_2/L)$, etc.) so that there is no prescribed boundary condition for $f_0$. We shall again compare the CF and FE approaches when $f_0$ is generated with different values of $s_0$ and $\kappa_0$. 

The exact procedure for the CF and FE approaches will be completely analogous to the 1D case. In particular, the same parameters $s_0$ and $\kappa_0$ that generate $f_0$ are used and furthermore the FE approach is carried out over the larger domain $[-\rho, L+\rho]^2$ with $\rho=\sqrt{8\nu}/\kappa$ to reduce boundary effects. However we remark that in the 2D case extending the domain has a larger impact on the performance of the FE approach than the 1D case. The reason is that to achieve the same mesh size within the domain $[0,L]^2$, the FE approach over $[-\rho,L+\rho]^2$ will require many more grid points than over $[0,L]^2$. In particular if a uniform partition of mesh size $h$ as in Subsection \ref{sec:fem hyper rectangle} is adopted, then the overall increment of number of grid points is 
\begin{align}
    \Big(\frac{L+2\rho}{h}\Big)^2-\Big(\frac{L}{h}\Big)^2=\Big(\frac{2\rho}{h}\Big)\Big(\frac{2L}{h}\Big) + \Big(\frac{\rho}{h}\Big)^2. \label{eq:increment grid}
\end{align}
The factor $2L/h$ makes \eqref{eq:increment grid} much larger than the increment $2\rho/h$ in each dimension and leads to a much larger saturation threshold $n_h^{\ast}$ (that we have introduced in Subsection \ref{sec:ex1}). For this reason we consider an alternate FE approach carried out over the smaller domain $[-0.1\rho,L+0.1\rho]^2$ and compare its performance with the other FE approach over $[-\rho,L+\rho]^2$. For the simulations that we are going to present, we fix $N=5000$, $L=5$, $\tau=0.1\cdot \|\mathbf{f}_0\|_2/\sqrt{N}$ and vary $s_0\in \{2,3,4\}$, $\kappa_0\in\{1,5,25\}$, where we recall $\mathbf{f}_0=(f_0(\X_1),\ldots,f_0(\X_N))^\top$. Similar parameter settings were considered in \cite{bolin2013comparison}.

Figure \ref{figure:ex2} shows qualitatively similar results as those in Subsection \ref{sec:ex1}, where the estimation error of both FE approaches decreases to that of the CF approach after certain threshold $n_h^{\ast}$. Although Theorem \ref{thm:rectangle} suggests a smaller asymptotic scaling for $n_h^{\ast}$ only when $s>3$, the simulation results suggest that this is true for $s=3$ and even for $s=2$ when the FE approach is computed over $[-0.1\rho,L+0.1\rho]^2$. Furthermore, no estimation accuracy is lost when this smaller domain is used and a smaller $n_h^{\ast}$ suffices so that it is more favorable, especially when $\kappa$ is small or equivalently when $\rho$ is large. Finally we remark that $s_0=2, \kappa_0=25$ corresponds to a very rapidly changing field and even in this case we have $n_h^{\ast}<N$ when $N=5000$, which is also a realistic amount of data relative to the domain size. Following the same intuition as provided in Figure \ref{figure:ex1}, it is reasonable to expect that one can take $n_h$  an order of magnitude smaller than $N$ when e.g. $N=50000$. Therefore we believe the results in Theorem \ref{thm:rectangle} have practical implications for a wide range of moderate nonasymptotic regimes and can provide some meaningful insights for real world applications.

\section{Discussion and Open Directions}\label{sec:conclusion}

In this paper we have employed a Bayesian nonparametrics framework to provide new understanding on the choice of the dimension $n_h$ in FE approaches to GP regression and classification. Our theory and simulation studies  demonstrate that under mild smoothness assumptions one can take $n_h\ll N$ for a wide range of practical scenarios without hindering the estimation accuracy, leading to a second layer of computational gain on top of the well-celebrated sparsity provided by the FE approach. 

One of the key elements in our analysis is the framework \cite{van2008rates} which allows to translate prior approximation guarantees to the posteriors.   
In the context of GP regression and classification, this boils down to controlling, respectively, the error $\mathbb{E}\|u_{h_N}-u\|_{\infty}^2$ and $\mathbb{E}\|u_{h_N}-u\|^2_2$ as in Proposition \ref{prop:prior approx implies same pc rate}. In Sections \ref{sec:main} and \ref{sec:numerics} we have used this framework to analyze the Mat\'ern covariance approach and linear FE approximations thereof on a hyperrectangle. We remark that the applicability of our framework goes beyond this simple setting with the following possible extensions.  

\paragraph{General Elliptic Operators}
One can define a nonstationary Mat\'ern-type GP similarly as in \eqref{eq:spde global} by replacing the operator $\kappa^2-\Delta$ with a more general elliptic operator $\kappa^2-\nabla \cdot (\bf{H}\nabla)$ on a general domain, where $\kappa$ and $\bf{H}$ are smooth functions taking values in real numbers and matrices, respectively. The error analysis in \cite{bolin2020numerical,cox2020regularity} on FE approximations of these random fields together with Proposition \ref{prop:prior approx implies same pc rate} would give a sufficient scaling of the mesh size $h_N$. We remark that in this case there would not be an easily computable covariance function approach to compare with, but one can still arrive at the conclusion that there is no need to discretize beyond the threshold implied by  Proposition \ref{prop:prior approx implies same pc rate}.

\paragraph{Higher Order FEM}
Higher order finite elements may be employed when the smoothness parameter $s$ is large. In particular, results from \cite{bolin2020numerical,cox2020regularity} show that the $L^2$ approximation rate in Lemma \ref{lemma:infty field rate rectangle} can be improved to $h^{(2s-D)\wedge (2p+2)}$ when polynomials of order $p$ are used. As a result, the scaling for $h_N$ in Theorem \ref{thm:rectangle} can be improved accordingly. The $L^{\infty}$ approximation rate for $p>1$ remains an interesting open question.

\paragraph{Rational Approximation}
When the smoothness parameter $s$ is not an integer, the favorable sparsity of the FE approximation is lost. For this reason, \cite{bolin2020rational} proposed a rational approximation of the fractional operator which retains sparsity. The resulting approximate field is shown to satisfy a similar $L^2$ approximation error bound as in Lemma \ref{lemma:infty field rate rectangle}, with an additional term coming from the rational approximation which can be made as small as desired. Following a similar argument as in \cite{bolin2020rational}, our proof for the $L^{\infty}$ bound can also be extended to the rational approximate field. These prior approximation rates can again be combined with Proposition \ref{prop:prior approx implies same pc rate} to yield a sufficient scaling of the mesh size.

\paragraph{Learning the Lengthscale}
A novel aspect of our error bounds for FE prior representations is that we keep track of the (inverse) lengthscale parameter $\kappa.$  Our theory and numerical experiments  help explain the need of finer discretization when the lengthscale is shorter. An interesting direction for further research is the design of algorithms for the simultaneous learning of (i) adaptive FE meshes for GP representations; and (ii)  spatially-variable lengthscale parameters $\kappa(x)$ in nonstationary Mat\'ern-type models. 

\paragraph{Beyond Regression and Classification} Lastly, we also envision that the framework we have introduced can be adopted in other problems such as density estimation \cite{ghosal2000convergence} and nonlinear Bayesian inverse problems \cite{stuart2010inverse}. The results from \cite{van2008rates} readily extend our framework to density estimation problems and it is an interesting direction to extend Proposition \ref{prop:prior approx implies same pc rate} to Bayesian inverse problem settings building, for instance, on \cite{nickl2020convergence,giordano2020consistency}.

\section{Proof of Main Results}\label{sec:proofs}

\begin{proof}[Proof of Proposition \ref{prop:covariance of FEM approximation}]
It suffices to find the coordinates of \eqref{eq:KL fem approx} in terms of the finite element basis $e_{h,i}$'s. Taking inner product of \eqref{eq:KL fem approx} with $e_{h,j}$, we get 
\begin{align*}
    \sum_{i=1}^{n_h}(\kappa^2 +\lambda_{h,i})^{-s/2}\xi_i\langle \psi_{h,i},e_{h,j} \rangle = \sum_{i=1}^{n_h}w_i\langle e_{h,i},e_{h,j}\rangle  \quad \quad j=1,\ldots,n_h,
\end{align*}
and the system 
\begin{align}
    \R(\kappa^2 \I_{n_h}+\bLambda)^{-s/2} \boldsymbol{\xi} = \M\w, \label{eq:covariance of FEM approximation system}
\end{align}
where   $\R_{ij}=\langle e_{h,i},\psi_{h,j}\rangle$   and  $\bLambda$ is the diagonal matrix with entries $\bLambda_{ii}=\lambda_{h,i}$. It now remains to relate $\bLambda$ with the matrices $\R, \M, \G$. Since the $\psi_{h,i}$'s form an orthonormal basis, we have $e_{h,i}=\sum_{j=1}^{n_h}\langle e_{h,i},\psi_{h,j}\rangle \psi_{h,j}$ $=\sum_{j=1}^{n_h} \R_{ij} \psi_{h,j}$, which implies that $\psi_{h,i}=\sum_{j=1}^{n_h} ( \R^{-1})_{ij} e_{h,j}$. The fact that $\psi_{h,i}$'s are (variational) eigenvectors of $-\Delta_h$ with corresponding eigenvalues $\lambda_{h,i}$ gives $\G\R^{-\top}= \M\R^{-\top} \bLambda$, which together with the fact that $\R \R^\top= \M$ further implies $\bLambda= \R^{-1}\G \R^{-\top}$. The result then follows by plugging such representation for $\bLambda$ into \eqref{eq:covariance of FEM approximation system}.
\end{proof}

\begin{proof}[Proof of Lemma \ref{lemma:spectral bound 1d}]
First note that we have $\lambda_0=0$, $\psi_0\equiv 1/L$ and
\begin{align*}
    \lambda_i=\left(\frac{i\pi}{L}\right)^2, \quad \psi_i=\frac{2}{L}\cos\left(\frac{i\pi x}{L}\right),\quad
    i=1,2,\ldots 
\end{align*}
We then have 
\begin{align*}
    |\lambda_{h,i}-\lambda_i|&\leq \frac{1}{2+\cos(i\pi h/L)}\left|\frac{6}{h^2}\left(1-\cos\frac{i\pi h}{L}\right)-\left(\frac{i\pi}{L}\right)^2\left(2+\cos \frac{i\pi h}{L}\right)\right|\\
    &\leq \left|\frac{6}{h^2}\left(1-\cos\frac{i\pi h}{L}\right)-\left(\frac{i\pi}{L}\right)^2\left(2+\cos \frac{i\pi h}{L}\right)\right|. 
\end{align*}
Expanding the last expression based on the Taylor series of $\cos x$ we obtain 
\begin{align*}
    |\lambda_{h,i}-\lambda_i| \leq \left|\frac{3}{h^2}\left(\frac{i\pi h}{L}\right)^2+\frac{1}{h^2}O\left(\frac{i\pi h}{L}\right)^4-3\left(\frac{i\pi}{L}\right)^2+\left(\frac{i\pi}{L}\right)^2O\left(\frac{i\pi h}{L}\right)^2\right|
    &\leq C\lambda_i^2 h^2.
\end{align*}
For the approximation error of the eigenfunctions, we first compute the normalizing constants $c_i$'s. For $i=0$ we notice that $\psi_{h,0}$ is constant and hence $c_0=1/L$. For general $i$'s, we denote $\psi_{h,i}=\sum_{k=0}^K z_{h,i,k}e_{h,k}$ and compute 
\begin{align*}
    \langle\psi_{h,i},\psi_{h,j}\rangle&=\sum_{k=0}^{K}\sum_{\ell=0}^{K} z_{h,i,k}z_{h,j,\ell}\langle e_{h,k},e_{h,\ell}\rangle\\
    &=z_{h,i,0}\sum_{\ell=0}^{K}z_{h,j,\ell}\langle e_{h,0},e_{h,\ell}\rangle+z_{h,i,K}\sum_{\ell=0}^{K}z_{h,j,\ell}\langle e_{h,K},e_{h,\ell}\rangle
    +\sum_{k=1}^{K-1} z_{h,i,k}\sum_{\ell=0}^{K}z_{h,j,\ell}\langle e_{h,k},e_{h,\ell}\rangle\\
    &=z_{h,i,0}\left(\frac{z_{h,j,0}}{3}+\frac{z_{h,j,1}}{6}\right)h+z_{h,i,K}\left(\frac{z_{h,j,K}}{3}+\frac{z_{h,j,K-1}}{6}\right)h \\
    &\quad +\sum_{k=1}^{K-1} z_{h,i,k} \left(\frac{z_{h,j,k-1}}{6}+\frac{2z_{h,j,k}}{3}+\frac{z_{h,j,k+1}}{6}\right)h\\
    &=c_ic_jh\left(\frac{1}{6}\cos \frac{j\pi h}{L}+\frac13\right)\left[1+(-1)^{i+j}+2\sum_{k=1}^{K-1} \cos\left(\frac{ki\pi h}{L}\right)\cos\left(\frac{kj \pi h}{L}\right)\right],
\end{align*}
where we have used  that $\cos(a-t)+\cos(a+t)=2\cos(a)\cos(t)$. Using further the fact that $2\cos(a)\cos(b)=\cos(a+b)+\cos(a-b)$, we have 
\begin{align*}
    2\sum_{k=1}^{K-1} \cos\left(\frac{ki\pi h}{L}\right)\cos\left(\frac{kj \pi h}{L}\right)=\sum_{k=1}^{K-1} \cos\left(\frac{k(i+j)\pi h}{L}\right) +\cos\left(\frac{k(i-j)\pi h}{L}\right). 
\end{align*}
Now letting $t=(i+j)\pi h/L$ and denoting $\iota$ as the imaginary unit, we have  
\begin{align*}
    \sum_{k=1}^{K-1} \cos(kt)=\frac12 \sum_{k=1}^{K-1} e^{\iota kt}+e^{-\iota kt}
    &=\frac12 \left[\frac{e^{\iota t}(1-e^{\iota (K-1) t})}{1-e^{\iota t}}+\frac{e^{-\iota t}(1-e^{-\iota (K-1) t})}{1-e^{-\iota t}}\right]\\
    &=-\frac12\left[1+ (-1)^{i+j}\right],
\end{align*}
where we have used  that $K t=(i+j)\pi$. Similarly, we have 
\begin{align*}
    \sum_{k=1}^{K-1} \cos\left(\frac{k(i-j)\pi h}{L}\right)=\begin{cases}
    -\frac12\left[1+ (-1)^{i-j}\right]\quad &i\neq j\\
    K-1 \quad &i=j
    \end{cases}.
\end{align*}
Therefore we have 
\begin{align*}
    \langle \psi_{h,i},\psi_{h,j}\rangle=c_ic_j L \left(\frac{1}{6}\cos \frac{j\pi h}{L}+\frac13\right) \delta_{ij}, 
\end{align*}
where $\delta_{ij}$ denotes the Kronecker delta. Therefore the $\psi_{h,i}$'s are orthonormal with 
\begin{align*}
    c_i=\left[ L \left(\frac{1}{6}\cos \frac{i\pi h}{L}+\frac13\right)\right]^{-1/2}. 
\end{align*}
Now to bound the eigenfunction approximation error, we have 
\begin{align*}
    \|\psi_{h,i}-\psi_i\|_{\infty}\leq \|\psi_{h,i}-\widetilde{\psi}_{h,i}\|_{\infty}+\|\widetilde{\psi}_{h,i}-\psi_i\|_{\infty},
\end{align*}
where $\widetilde{\psi}_{h,i}=\frac{1}{c_i}\sqrt{\frac{2}{L}}\psi_{h,i}.$ Since $\|\psi_{h,i}\|_{\infty}\leq 1$, we have 
\begin{align*}
    \|\psi_{h,i}-\widetilde{\psi}_{h,i}\|_{\infty}\leq \left|1-\frac{1}{c_i}\sqrt{\frac{2}{L}}\right|\leq C\left|L\left(\frac{1}{6}\cos \frac{i\pi h}{L}+\frac13\right)-\frac{L}{2}\right|\leq C\left(\frac{i\pi h}{L}\right)^2=C\lambda_ih^2,
\end{align*}
where $C$ is constant depending only on $L$. To bound $\|\widetilde{\psi}_{h,i}-\psi_i\|_{\infty}$, notice that after the rescaling, $\widetilde{\psi}_{h,i}$ is a linear interpolant of $\psi_i$ over the nodes. In particular, denoting $x_k=kh$ we have 
\begin{align*}
    \widetilde{\psi}_{h,i}(x)=\psi_i(x_k) + \frac{x-x_k}{x_{k+1}-x_k} [\psi_i(x_{k+1})-\psi_i(x_k)]   \quad \quad \operatorname{on}\,\, [x_k,x_{k+1}].
\end{align*}
Taylor expanding at $x$ we have 
\begin{align*}
    \psi_i(x_{k+1})&=\psi_i(x)+\psi_i^{\prime}(x)(x_{k+1}-x)+\frac{\psi_i^{\prime\prime}(\eta_1)}{2}(x_{k+1}-x)^2 \quad \quad x<\eta_1<x_{k+1}\\
    \psi_i(x_k)&=\psi_i(x)+\psi_i^{\prime}(x)(x_k-x)+\frac{\psi_i^{\prime\prime}(\eta_2)}{2}(x_k-x)^2 \quad \quad x_k<\eta_2<x.
\end{align*}
Therefore 
\begin{align*}
    \widetilde{\psi}_{h,i}(x)=\psi_i(x)+\frac{\psi_i^{\prime\prime}(\eta_2)}{2}(x_k-x)^2+\frac{x-x_k}{x_{k+1}-x_k}\left[\frac{\psi_i^{\prime\prime}(\eta_1)}{2}(x_{k+1}-x)^2+\frac{\psi_i^{\prime\prime}(\eta_2)}{2}(x_k-x)^2\right]
\end{align*}
and since $\|\psi_i^{\prime\prime}\|_{\infty}\leq \lambda_i$ we have
\begin{align*}
    \underset{x\in[x_k,x_{k+1}]}{\operatorname{sup}}\,|\widetilde{\psi}_{h,i}(x)-\psi_i(x)| \leq C \lambda_i h^2. 
\end{align*}
Therefore $\|\widetilde{\psi}_{h,i}-\psi_i\|_{\infty}\leq C\lambda_ih^2$ and the result follows.
\end{proof}

\begin{proof}[Proof of Lemma \ref{lemma:rectangle FEM eigenfucntion infty bound}]
Notice that 
\begin{align*}
    |\Lambda_{\h,\i}-\Lambda_{\i}|\leq \sum_{d=1}^D |\lambda_{h_d,i_d}-\lambda_{i_d}|
\end{align*}
and 
\begin{align*}
    |\Psi_{\h,\i}(\x)-\Psi_{\i}(\x)|&=\left|\prod_{d=1}^D\psi_{h_d,i_d}(x_d)-\prod_{d=1}^D\psi_{i_d}(x_d)\right|\nonumber\\
    &\leq \sum_{d=1}^D\bigg( |\psi_{h_d,i_d}(x_d)-\psi_{i_d}(x_d)| \prod_{\ell=1}^{d-1} \left|\psi_{i_{\ell}}(x_{\ell})\right|  \prod_{\ell=d+1}^D |\psi_{h_{\ell},i_{\ell}}(x_{\ell})|\bigg).
\end{align*}
Therefore the result follows from the one-dimensional estimates in Lemma \ref{lemma:spectral bound 1d}. 
\end{proof}

\begin{proof}[Proof of Lemma \ref{lemma:infty field rate rectangle}]
We shall abuse the notation and order the multi index $\i\in \mathbb{N}^D$ as a single sequence $i\in\mathbb{N}$ so that  
\begin{align*}
    u&=\kappa^{s-D/2}\sum_{i=1}^{\infty} (\kappa^2+\Lambda_i)^{-\frac{s}{2}}\xi_i\Psi_i  \, ,\\
    u_{\h}&=\kappa^{s-D/2}\sum_{i=1}^{n_{\h}}(\kappa^2+\Lambda_{\h,i})^{-\frac{s}{2}}\xi_i\Psi_{\h,i} \, .
\end{align*}

\paragraph{Bound for $\mathbb{E}\|u_{\h}-u\|_2^2$}   This can be proven using the techniques in \cite[Theorem 2.10]{bolin2020numerical} and in the $L^{\infty}$ bound that we will establish below.

\paragraph{H\"older continuity}
We have 
\begin{align*}
    \mathbb{E}|u(\x+\h)-u(\x)|^2&\lesssim  \sum_{i=1}^{\infty}(\kappa^2+\Lambda_i)^{-s}|\Psi_i(\x+\h)-\Psi_i(\x)|^2\\
    &\lesssim \sum_{i=1}^{\infty}(\kappa^2+\Lambda_i)^{-s} \operatorname{min}\{\|\nabla \Psi_i\|_{\infty}^2h^2,1\}\\
    & \lesssim  \sum_{i=1}^{\infty}i^{-\frac{2s}{D}} \operatorname{min}\{i^{\frac{2}{D}}h^2,1\}\\
    & \lesssim \int_{x\geq 1} x^{-\frac{2s}{D}} \operatorname{min}\{x^{\frac{2}{D}}h^2,1\}dx\\
    & \lesssim  \left[\int_{1\leq x\leq h^{-D}}h^2x^{-\frac{2}{D}-\frac{2s}{D}} dx +\int_{x>h^{-D}} x^{-\frac{2s}{D}}dx\right]\\
    & \lesssim h^{2s-D},
\end{align*}
where we have used that $\|\nabla \Psi_i\|_{\infty}\leq \sqrt{\Lambda_i}$ and Weyl's law $\Lambda_i\asymp i^{2/D}$. Then by \cite[Corollary 6.8]{stuart2010inverse} we have 
\begin{align*}
    \mathbb{E}|u(\x)-u(\x')|^{2p}\leq C_p |\x-\x'|^{(2s-D)p}
\end{align*}
for all $p\in\mathbb{N}$. Kolmogorov continuity theorem \cite[Theorem 6.24]{stuart2010inverse} implies that $u$ is $\beta$-H\"older for $\beta<\frac{(2s-D)p-2D}{2p}$. Letting $p\rightarrow\infty$ gives the desired result.

\paragraph{Bound for $\mathbb{E}\|u_{\h}-u\|_{\infty}^2$} Consider two intermediate quantities
\begin{align}
    \tilde{u}&=\kappa^{s-D/2}\sum_{i=1}^{n_{\h}}(\kappa^2+\Lambda_i)^{-\frac{s}{2}}\xi_i\Psi_i \, , \label{eq:intermediate 1}\\
    \tilde{u}_{\h}&=\kappa^{s-D/2}\sum_{i=1}^{n_{\h}}(\kappa^2+\Lambda_i)^{-\frac{s}{2}}\xi_i\Psi_{\h,i} \, .\label{eq:intermediate 2}
\end{align}
We have 
\begin{align*}
    \mathbb{E}\|u-u_{\h}\|_{\infty}^2&\leq \mathbb{E}\left(\|u-\tilde{u}\|_{\infty}+\|\tilde{u}-\tilde{u}_{\h}\|_{\infty}+\|\tilde{u}_{\h}-u_{\h}\|_{\infty}\right)^2\\
    &\leq 2\left(\mathbb{E}\|u-\tilde{u}\|_{\infty}^2 +\mathbb{E}\|\tilde{u}-\tilde{u}_{\h}\|_{\infty}^2 +\mathbb{E}\|\tilde{u}_{\h}-u_{\h}\|_{\infty}^2\right)
\end{align*}
and it suffices to bound each term. 
Since the $\Psi_i$'s are uniformly bounded and that $\xi$ has bounded first moment, we have  
\begin{align}
    \mathbb{E}\|u-\tilde{u}\|_{\infty}^2 &\leq  \kappa^{2s-D} \mathbb{E}\left[\sum_{i=n_{\h}+1}^{\infty}(\kappa^2+\Lambda_i)^{-\frac{s}{2}}|\xi_i|\|\Psi_i\|_{\infty}\right]^2\nonumber \\
    & \lesssim \kappa^{2s-D} \left[\sum_{i=n_{\h}+1}^{\infty}(\kappa^2+\Lambda_i)^{-\frac{s}{2}}\right]^2\nonumber\\
    &\lesssim \kappa^{2s-D} \left(\sum_{i=n_{\h}+1}^{\infty} i^{-\frac{s}{D}}\right)^2 
    \lesssim \kappa^{2s-D} n_{\h}^{2-\frac{2s}{D}}\asymp \kappa^{2s-D}h^{2s-2D}. \label{eq:e1 infinity error rectangle}
\end{align}
Similarly, by Lemma \ref{lemma:rectangle FEM eigenfucntion infty bound}
\begin{align}
    \mathbb{E} \|\tilde{u}-\tilde{u}_{\h}\|_{\infty}^2 
    &\lesssim\kappa^{2s-D}\left[\sum_{i=1}^{n_{\h}}(\kappa^2+\Lambda_i)^{-\frac{s}{2}}\|\Psi_i-\Psi_{\h,i}\|_{\infty}\right]^2 \nonumber\\
    &\lesssim \kappa^{2s-D}h^4\left(\sum_{i=1}^{n_{\h}}\Lambda_i^{1-\frac{s}{2}}\right)^2 \lesssim \kappa^{2s-D}h^4\left[1\vee n_{\h}^{2+\frac{4}{d}\left(1-\frac{s}{2}\right)}
    \right]
    \asymp \kappa^{2s-D}h^{(2s-2s)\wedge 4}.\label{eq:e2 infinity error rectangle}
\end{align}
For the last term we have by Lemma \ref{lemma:rectangle FEM eigenfucntion infty bound}
\begin{align}
    \mathbb{E}\|\tilde{u}_{\h}-u_{\h}\|_{\infty}^2&\lesssim \kappa^{2s-D} \left[\sum_{i=1}^{n_{\h}} \left|(\kappa^2+\Lambda_i)^{-\frac{s}{2}}-(\kappa^2+\Lambda_{\h,i})^{-\frac{s}{2}}\right|\right]^2\nonumber \\
    & \lesssim \kappa^{2s-D}\left[\sum_{i=1}^{n_{\h}}\Lambda_i^{-\frac{s}{2}-1}|\Lambda_i-\Lambda_{h,i}|\right]^2 
    \lesssim \kappa^{2s-D}h^4 \left(\sum_{i=1}^{n_{\h}}\Lambda_i^{1-\frac{s}{2}}\right)^2
    \lesssim \kappa^{2s-D} h^{(2s-2d)\wedge 4}. \label{eq:e3 infinity error rectangle}
\end{align}
The result follows by combining \eqref{eq:e1 infinity error rectangle}, \eqref{eq:e2 infinity error rectangle}, \eqref{eq:e3 infinity error rectangle}.
\end{proof}

\begin{proof}[Proof of Proposition \ref{lemma:besov example infnity rectangle}]
Again we shall abuse the notation and write
\begin{align*}
    u&=\kappa^{s-D/2}\sum_{i=1}^{\infty} (\kappa^2+\Lambda_i)^{-s/2}\xi_i\Psi_i.
\end{align*}
\paragraph{$L^{\infty}$ case}
Recall that 
\begin{align*}
    \phi_{f_0}(\eps; u,\|\cdot\|_{\infty})= \underset{g\in\mathbb{H}:\|g-f_0\|_{\infty}<\eps}{\operatorname{inf}}\, \|g\|^2_{\mathbb{H}}- \log \mathbb{P}(\|u\|_{\infty}<\eps). 
\end{align*}
By \cite[Theorem 1.2]{li1999approximation}, the second term can be bounded by analyzing the $L^{\infty}(\mathcal{D})$ metric entropy of $\mathbb{H}_1$, the unit ball of $\mathbb{H}$. Notice that $\mathbb{H}_1$ takes the form
\begin{align*}
    \mathbb{H}_1=\left\{\sum_{i=1}^{\infty}g_i\Psi_i : \sum_{i=1}^{\infty}g_i^2(\kappa^2+\Lambda_i)^s\leq 1 \right\}
\end{align*}
and is contained in a Sobolev ball of order $s$, whose $L^{\infty}(\mathcal{D})$ metric entropy is bounded by a constant times $\eps^{-\frac{D}{s}}$ (see e.g. \cite[Theorem 3.3.2]{edmunds1996function}). Then \cite[Theorem 1.2]{li1999approximation} implies that 
\begin{align}
    - \log \mathbb{P}(\|u\|_{\infty}<\eps) \lesssim  \eps^{-\frac{2D}{2s-D}}=\eps^{-\frac{D}{\beta}}, \label{eq:small ball probability infinity rectangle}
\end{align}
where we used the assumption that $s=\beta+\frac{D}{2}$. 
For the first term, let $C_0=\|f_0\|_{B_{\infty,\infty}^{\beta}}$ and consider $g=S_J(\sqrt{\Delta})f_0$ with $J$ the smallest integer such that $C_02^{-\beta J}<\varepsilon$. 
Since $f_0\in B_{\infty,\infty}^{\beta}$ we have  
\begin{align}
    \|S_j(\sqrt{\Delta})f_0-f_0\|_{\infty}\leq C_0 2^{-\beta j} \label{eq:approx rate f_0 infinity rectangle}
\end{align}
for all $j$.  
In particular, $\|g-f_0\|_{\infty}\leq C_02^{-\beta J}\leq \eps$. Moreover we have $g=\sum_{i=1}^{\infty}S_J(\sqrt{\Lambda_i})f_i\Psi_i$ as a finite series since $S_J(\sqrt{\Lambda_i})=0$ if $\sqrt{\Lambda_i}> 2^J$ and hence $g\in\mathbb{H}$. Now since $S_j\leq 1$,
\begin{align*}
    \|g\|_{\mathbb{H}}^2\leq \sum_{\sqrt{\Lambda_i}\leq 2^{J}} f_i^2(\kappa^2+\Lambda_i)^s
    = \sum_{\sqrt{\Lambda_i}\leq 1} f_i^2(\kappa^2+\Lambda_i)^s + \sum_{j=1}^J \sum_{2^{j-1}< \sqrt{\Lambda_i}\leq 2^j}f_i^2(\kappa^2+\Lambda_i)^s.
\end{align*}
By \eqref{eq:approx rate f_0 infinity rectangle} we have 
\begin{align*}
    \sum_{2^{j-1}< \sqrt{\Lambda_i}\leq 2^j}f_i^2(\kappa^2+\Lambda_i)^s \lesssim 2^{2js} \sum_{2^{j-1}< \sqrt{\Lambda_i}}f_i^2 &\lesssim 2^{2js}\|S_j(\sqrt{\Delta})f_0-f_0\|_{2}^2 \\
    &\lesssim 2^{2js}\|S_j(\sqrt{\Delta})f_0-f_0\|_{\infty}^2\lesssim 2^{2(s-\beta)j}.
\end{align*}
Since $J$ is the smallest integer such that $C_02^{-\beta J}<\eps$, we have $2^{-\beta J}\gtrsim \eps$ and hence 
\begin{align}
    \|g\|_{\mathbb{H}}^2 \lesssim \sum_{\sqrt{\Lambda_i}\leq 1} f_i^2 + \sum_{j=1}^J 2^{2(s-\beta)j} \lesssim 2^{2(s-\beta)J}\lesssim \eps^{-\frac{2(s-\beta)}{\beta}}=\eps^{-\frac{D}{\beta}}. \label{eq:decentering function infinity rectangle}
\end{align}
Combining \eqref{eq:small ball probability infinity rectangle} and \eqref{eq:decentering function infinity rectangle} we deduce that
\begin{align*}
    \phi_{f_0}(\eps)\lesssim \eps^{-\frac{D}{\beta}}
\end{align*}
and setting $\eps_n=C n^{-\frac{\beta}{2\beta+D}}$ for a large enough constant $C$ gives the result.

\paragraph{$L^2$ case}We have 
\begin{align}
    \phi_{f_0}(\eps; u,\|\cdot\|_2)= \underset{g\in\mathbb{H}:\|g-f_0\|_2<\eps}{\operatorname{inf}}\, \|g\|^2_{\mathbb{H}}- \log \mathbb{P}(\|u\|_2<\eps). \label{eq:concentration function L^2}
\end{align}
For the second term, recall that $u=\kappa^{s-D/2}\sum_{i=1}^{\infty}(\kappa^2+\Lambda_i)^{-s/2}\xi_i\Psi_i$. We then have 
\begin{align}
    \log \mathbb{P}(\|u\|_2<\eps)&
    =\log \mathbb{P}\left(\kappa^{2s-D}\sum_{i=1}^{\infty}(\kappa^2+\Lambda_i)^{-s}\xi_i^2<\eps\right)\nonumber\\
    &\geq \log \mathbb{P}\left(\sum_{i=1}^{\infty}i^{-\frac{2s}{D}}\xi_i^2<C\eps\right)\gtrsim \eps^{-\frac{2}{\frac{2s}{D}-1}}=\eps^{-\frac{D}{\beta}},\label{eq:small ball probability L^2}
\end{align}
where the last step follows from \cite[Corollary 6]{dunker1998small} and the assumption that $s=\beta+\frac{D}{2}$. 
For the first term in \eqref{eq:concentration function L^2}, let $C_0=\|f_0\|_{B_{\infty,\infty}^{\beta}}$ and consider $g=S_J(\sqrt{\Delta})f_0$ with $J$ the smallest integer so that $C_0\sqrt{|\mathcal{D}|}2^{-\beta J}<\varepsilon$, where $|\mathcal{D}|$ is the Lebesgue measure of $\mathcal{D}$. Since $f_0\in B_{\infty,\infty}^{\beta}$ we have  
\begin{align*}
    \|S_j(\sqrt{\Delta})f_0-f_0\|_{\infty}\leq C_0 2^{-\beta j} 
\end{align*}
for all $j$. In particular, $\|g-S_J(\sqrt{\Delta})f_0\|_2\leq \sqrt{|\mathcal{D}|}\|g-S_J(\sqrt{\Delta})f_0\|_{\infty}\leq C_0\sqrt{|\mathcal{D}|}2^{-\beta J}\leq \eps$. Now proceeding in the same way as the argument in the $L^{\infty}$ case we obtain 
\begin{align}
    \|g\|_{\mathbb{H}}^2 \lesssim \sum_{\sqrt{\lambda_i}\leq 1} f_i^2 + \sum_{j=1}^J 2^{2(s-\beta)j} \lesssim 2^{2(s-\beta)J}\lesssim \eps^{-\frac{2(s-\beta)}{\beta}}=\varepsilon^{-\frac{D}{\beta}},\label{eq:decentering function L^2}
\end{align}
where we have used the fact that $2^{-\beta J}\gtrsim \eps$ since $J$ is the smallest integer such that $C_0\sqrt{|\mathcal{D}|}2^{-\beta J}$ $<\varepsilon$. Combining \eqref{eq:small ball probability L^2} and \eqref{eq:decentering function L^2} we deduce that
\begin{align*}
    \phi_{f_0}(\eps)\lesssim \eps^{-\frac{D}{\beta}},
\end{align*}
and setting $\eps_n=C n^{-\frac{\beta}{2\beta+D}}$ for a large enough constant $C$ gives the result.
\end{proof}

\section*{Acknowledgements}  
  Both authors are thankful for the support of NSF and NGA through the grant DMS-2027056. DSA is also supported by a Fundación BBVA start-up grant.  The authors are grateful to Ridgway Scott for helpful discussions. 

\bibliographystyle{abbrvnat} 
\bibliography{references}

\end{document}